\documentclass{article}
\usepackage{ijcai25}

\usepackage{times}
\usepackage{soul}
\usepackage{url}
\usepackage[hidelinks]{hyperref}
\usepackage[utf8]{inputenc}
\usepackage[small]{caption}
\usepackage{graphicx}
\usepackage{amsmath}
\usepackage{amsthm}
\usepackage{thmtools}
\renewcommand\thmcontinues[1]{Continued}
\usepackage{booktabs}
\usepackage{algorithm}
\usepackage{algorithmic}
\usepackage[switch]{lineno}

\usepackage{fontawesome}

\title{
Verifying Quantized Graph Neural Networks is PSPACE-complete
}

\author{
\#4158
}

\author{
Marco Sälzer$^1$
\and
François Schwarzentruber$^2$\and
Nicolas Troquard$^{3}$
\affiliations
$^1$ Technical University of Kaiserslautern, Kaiserslautern, Germany \\
$^2$ENS de Lyon, CNRS, Université Claude Bernard Lyon 1, Inria, LIP, UMR 5668, 69342, Lyon, France \\
$^3$Gran Sasso Science Institute (GSSI), Viale F.~Crispi, 7 -- 67100 L'Aquila, Italy
\emails
marco.saelzer@rptu.de, 
francois.schwarzentruber@ens-lyon.fr,
nicolas.troquard@gssi.it
}

\usepackage{amsmath}
\usepackage{amsthm}
\usepackage{amssymb}
\usepackage{graphicx}
\usepackage{tikz}
\usepackage{pgfplots}

\usepackage{todonotes}

\usepackage{bussproofs}
\usepackage{cleveref}
\usepackage{yfonts}
\usepackage{ifthen}
\usepackage{adjustbox}
\usepackage{comment}
\usepackage{multicol}
\usepackage{fontawesome}
\usepackage{blkarray}
\usepackage{bussproofs}
\usepackage{xcolor}

\newcommand{\citet}[1]{\citeauthor{#1}~\shortcite{#1}}

\newcommand{\rulelabel}[1]{
\LeftLabel{$(#1)$}
}
\usetikzlibrary{decorations.markings}
\usetikzlibrary{calc,patterns,angles,quotes}

\newcommand\indexalert\emph
\newcommand\alert\emph

\newcommand{\changed}[1]{#1}

\protected\def\MSO{\ifmmode \mbox{\sc MSO} \else {\sc MSO}\xspace\fi}
\protected\def\FO{\ifmmode \mbox{\sc FO} \else {\sc FO}\xspace\fi}
\protected\def\MMSO{\ifmmode {\mbox{\sc M}\MSO} \else {{\sc M}\MSO}\xspace\fi}
\protected\def\MFO{\ifmmode \mbox{\sc MFO} \else {\sc MFO}\xspace\fi}

\newcommand{\union}{\cup}

\newcommand{\compl}[1]{{\overline{{#1}}}}

\newcommand{\set}[1]{{\{#1\}}}

\newcommand{\bigand}{\bigwedge}
\newcommand{\lbigand}{\bigand}

\newcommand{\lequiv}{\leftrightarrow}

\renewcommand{\phi}{\varphi}

\newcommand{\lbox}{\square}

\newcommand{\ldiamond}{\Diamond}

\newcommand{\setR}{\mathbb{R}}

\newcommand{\ensR}{\setR}

\newtheorem{theorem}{Theorem}
\newtheorem{example}[theorem]{Example}

\newtheorem{definition}[theorem]{Definition}
\newtheorem{corollary}[theorem]{Corollary}

\newcommand{\algofunction}{\textbf{fonction }}

\newcommand{\algoprocedure}{\textbf{procédure }}

\newcommand{\algofor}{\textbf{pour }}
\newcommand{\algodo}{\textbf{faire }}

\definecolor{algocommentbackgroundcolor}{rgb}{1,1,0.5}

\newcommand{\algowhile}{\textbf{while }}

\newcommand{\algoif}{\textbf{if }}
\newcommand{\algothen}{\textbf{then }}

\newcommand{\algomatch}{\textbf{match }}

\newcommand{\algocase}{\textbf{case }}

\newlength{\algoindentlongueur}
\newcommand{\algoindent}{\hspace*{\algoindentlongueur}}

\newlength{\algoindentavantvrulelongueur}
\newcommand{\algoindentavantvrule}{\hspace*{\algoindentavantvrulelongueur}}

\newlength{\dummy}

\newsavebox{\frameminipageboiteavecunnomsuperlongdesortequonnepuissepaslereutiliser}
\newenvironment{frameminipage}[2][c]{%
\begin{lrbox}{\frameminipageboiteavecunnomsuperlongdesortequonnepuissepaslereutiliser}%
\begin{minipage}[#1]{#2}%
} {%
\end{minipage}%
\end{lrbox}%
\framebox{\usebox{\frameminipageboiteavecunnomsuperlongdesortequonnepuissepaslereutiliser}}%
}

\newenvironment{algobloc}{\setlength{\dummy}{\linewidth}\addtolength{\dummy}{- \algoindentlongueur}\addtolength{\dummy}{- \algoindentavantvrulelongueur}\algoindentavantvrule\vrule\algoindent\begin{minipage}{\dummy}}{\end{minipage}}

\tikzstyle{zoneprogramcounter} = [fill=gray!40, draw=none]
\tikzstyle{zonedatacounter} = [fill=gray!40, draw=none, decorate, decoration={snake, amplitude=1pt, segment length=4pt}]
\tikzstyle{portion} = [densely dotted, shade, top color = white, bottom color = gray!40]
\tikzstyle{portiondata} = [densely dotted, shade, top color = white, bottom color = gray!40, decorate, decoration={snake, amplitude=1pt, segment length=4pt}]
\tikzset{
	double arrow/.style args={#1 colored by #2 and #3}{
		-stealth,line width=#1,#2, 
		postaction={draw,-triangle 90 cap,#3,line width=(#1)-4pt,
			shorten <=4pt,shorten >=4}, 
	}
}
\tikzstyle{tapearrow} = [double arrow=10pt colored by black!50!white and black!20!white, -triangle 90 cap, fill = white]%
\tikzstyle{execarrow} = [line width=1pt,->]%

\tikzstyle{world}=[inner sep=0.5mm]
\tikzstyle{event}=[fill=gray!30, inner sep=0.5mm]
\tikzstyle{realworldarrowfromleft} = [initial left, initial text={}]
\tikzstyle{realworldarrowfromright} = [initial right, initial text={}]

\protected\def\ZPP{\ifmmode \mbox{\sc ZPP} \else {\sc ZPP}\xspace\fi}
\protected\def\RP{\ifmmode \mbox{\sc RP} \else {\sc RP}\xspace\fi}
\protected\def\coRP{\ifmmode \mbox{\sc coRP} \else {\sc coRP}\xspace\fi}
\protected\def\DTIME{\ifmmode \mbox{\sc Dtime} \else {\sc Dtime}\xspace\fi}
\protected\def\NTIME{\ifmmode \mbox{\sc Ntime} \else {\sc Ntime}\xspace\fi}
\protected\def\DSPACE{\ifmmode \mbox{\sc Dspace} \else {\sc Dspace}\xspace\fi}
\protected\def\NSPACE{\ifmmode \mbox{\sc Nspace} \else {\sc Nspace}\xspace\fi}
\protected\def\NP{\ifmmode \mbox{\sc NP} \else {\sc NP}\xspace\fi}
\protected\def\coNP{\ifmmode \mbox{\sc coNP} \else {\sc coNP}\xspace\fi}
\protected\def\NPSPACE{\ifmmode \mbox{\sc NPspace} \else {\sc NPspace}\xspace\fi}
\protected\def\PSPACE{\ifmmode \mbox{\sc Pspace} \else {\sc Pspace}\xspace\fi}
\protected\def\EXPSPACE{\ifmmode \mbox{\sc Expspace} \else {\sc Expspace}\xspace\fi}
\protected\def\TWOEXPSPACE{\ifmmode \mbox{\sc 2Expspace} \else {\sc 2Expspace}\xspace\fi}
\protected\def\PTIME{\ifmmode \mbox{\sc P} \else {\sc P}\xspace\fi}
\protected\def\NPTIME{\ifmmode \mbox{\sc NP} \else {\sc NP}\xspace\fi}
\protected\def\EXPTIME{\ifmmode \mbox{\sc Exptime} \else {\sc Exptime}\xspace\fi}
\protected\def\AEXPTIME{\ifmmode \mbox{\sc Aexptime} \else {\sc Aexptime}\xspace\fi}
\protected\def\NEXPTIME{\ifmmode \mbox{\sc NExptime} \else {\sc NExptime}\xspace\fi}
\protected\def\2EXPTIME{\ifmmode \mbox{\sc 2-Exptime} \else {\sc
		2-Exptime}\xspace\fi}
\DeclareRobustCommand{\kEXPTIME}[1][k]{\ifmmode \mbox{\sc $#1$-Exptime}
	\else {\sc $#1$-Exptime}\xspace\fi}
\DeclareRobustCommand{\kNEXPTIME}[1][k]{\ifmmode \mbox{\sc $#1$-NExptime}
	\else {\sc $#1$-NExptime}\xspace\fi}
\DeclareRobustCommand{\kEXPSPACE}[1][k]{\ifmmode \mbox{\sc $#1$-Expspace}
	\else {\sc $#1$-Expspace}\xspace\fi}
\protected\def\ELEMENTARY{\ifmmode \mbox{\sc Elementary} \else {\sc Elementary}\xspace\fi}
\protected\def\AEXPpol{\ifmmode \mbox{{\sc A}_{\text{pol}}\EXPTIME} \else
	{\sc A}$_{\text{pol}}$\EXPTIME\fi}
\protected\def\APTIME{\ifmmode \mbox{\sc Aptime} \else {\sc Aptime}\xspace\fi}
\protected\def\AEXPSPACE{\ifmmode \mbox{\sc Aexpspace} \else {\sc Aexpspace}\xspace\fi}

\newcommand{\gloups}[1]{\bigcirc}

\tikzstyle{cell} = [draw,minimum height=5mm,minimum width=5mm]
\tikzset{
	cellcolor/.cd,
	0/.style={fill=gray!20!white},
	1/.style={fill=yellow!20!white}
}
\tikzstyle{cellalive} = [fill=yellow!20!white]

\newcommand{\labellednode}[3]{
$#1$
$\begin{pmatrix}
#2 \\
#3
\end{pmatrix}$
}

\newcommand{\labellednodeinv}[3]{
$\begin{pmatrix}
#2 \\
#3
\end{pmatrix}$
$#1$
}

\newcommand\listeventsempty\epsilon

\newcommand{\initialword}{\omega}

\definecolor{C}{rgb}{1,1,1}
\definecolor{C }{rgb}{0.8,0.8,0.8}

\definecolor{Ca}{rgb}{0.95,1,0.5}
\definecolor{Cb}{rgb}{1,0.9,0.5}
\definecolor{Cw_2}{rgb}{0.95,1,0.5}
\definecolor{Cw_n}{rgb}{1,0.9,0.5}

\definecolor{Cq'-}{rgb}{0.5,0.5,1}
\definecolor{C-q_0}{rgb}{0.5,0.5,1}
\definecolor{C-q_1}{rgb}{0.5,0.4,1}
\definecolor{C-q'}{rgb}{0.5,0.8,0.5}
\definecolor{Cq_f-}{rgb}{0.4,0.9,0.5}

\definecolor{Cq_0,w_1}{rgb}{1,0.6,0.5}
\definecolor{Cq_0,}{rgb}{1,0.6,0.4}
\definecolor{Cq_2,}{rgb}{0.8,0.9,0.4}
\definecolor{Cq_2,a}{rgb}{0.8,0.9,0.4}
\definecolor{Cq_f,a}{rgb}{0.7,0.9,0.4}
\definecolor{Cq_0,a}{rgb}{1,0.6,0.5}
\definecolor{Cq_0,b}{rgb}{1,0.5,0.5}
\definecolor{Cq_1,a}{rgb}{1,0.5,0.4}
\definecolor{Cq_1,b}{rgb}{1,0.4,0.4}
\definecolor{Cq,a}{rgb}{1,0.5,0.5}
\definecolor{Cq',a}{rgb}{1,0.5,0.5}
\definecolor{Cq',b}{rgb}{1,0.5,0.5}
\definecolor{Cq_0,1}{rgb}{0.5,0.5,0.5}

\provideboolean{student}

\newcommand{\trou}[1]{\ifthenelse{\boolean{student}}
	{\colorbox{gray!10!white}{\phantom{#1}}}
	{#1}
}

\ifthenelse{\boolean{student}}
{\excludecomment{hidden}}
{}

\usepackage{mdframed}

\newmdenv{allfour}

\newmdenv[rightline=false,topline=false,bottomline=false]{solutionbox}

\ifthenelse{\boolean{student}}
{\excludecomment{solution}}
{
}

\ifthenelse{\boolean{student}}
{\excludecomment{examplehidden}}
{}

\ifthenelse{\boolean{student}}
{\excludecomment{exercicehidden}}
{}

\ifthenelse{\boolean{student}}
{\excludecomment{remarkhidden}}
{}

\definecolor{proba}{RGB}{0, 0, 128}

\usepackage{tikzsymbols}

\newcommand{\taille}[1]{|#1|}
\newcommand\sizeof\taille

\newcommand\instancemot\initialword

\newcommand{\aGNN}{N}
\newcommand{\gmlbox}[1]{\exists^{\geq #1}}

\newcommand{\Ktwoa}{K_{2a}}
\newcommand{\fragmentGML}{\text{GML}_3}
\newcommand{\fragmentGMLmodallogic}{\text{GML}_1}

\newcommand{\comb}{\mathit{comb}}
\newcommand{\agg}{\mathit{agg}}
\newcommand{\graphs}{\mathcal{G}}
\newcommand{\pgraphs}{\mathcal{G}_\bullet}
\newcommand{\bs}[1]{\boldsymbol{#1}}
\newcommand{\multset}[1]{\{\!\{#1\}\!\}}

\newcommand{\thelogic}
{\ensuremath{\mathcal L_\textsf{quantGNN}}}

\newcommand{\featuresset}{F}

\newcommand{\multiset}[1]{ \{\{ #1 \}\} }

\newcommand{\aggvar}{a}

\newcommand{\expression}{\xi} 
\newcommand{\nodes}{V}
\newcommand{\edges}{E}
\newcommand{\activationfunction}{\alpha}
\newcommand{\agreggationfunction}{agg}

\newcommand{\layer}{\mathcal{L}}

\newtheorem{lemma}[theorem]{Lemma}
\newtheorem{remark}[theorem]{Remark}
\newcommand{\degree}{\delta}

\newcommand\citesuplementarymaterial{[Appendix]}

\begin{document}

\maketitle

\begin{abstract}
 
In this paper, we investigate the verification of quantized Graph Neural Networks (GNNs), where some fixed-width arithmetic is used to represent numbers.
We introduce the linear-constrained validity (LVP) problem for verifying GNNs properties, and provide an efficient translation from LVP instances into a logical language. We show that LVP is in PSPACE, for any reasonable activation functions. We provide a proof system. We also prove PSPACE-hardness, indicating that while reasoning about quantized GNNs is feasible, it remains generally computationally challenging.
\end{abstract}

\section{Introduction}
\label{sec:intro}

Graph Neural Networks (GNNs) are 
neural network models computing functions over graphs or graph-node pairs. Their ability to extract feature and structural information from
graphs
has made GNNs a promising candidate for 
 tackling tasks in fields such as social network analysis, chemistry applications, and knowledge graphs. For a comprehensive survey of 
GNN applications, see \citet{ZhouCHZYLWLS20}, and for a survey of GNNs for knowledge graphs, see \citet{YeKSSW22}.

The increasing use of GNNs (or neural network-based models in general) leads to a heightened necessity for reliable safety guarantees or explanations of their behavior, in order to meet legal requirements or client desiderata. Let us illustrate some guarantees we may expect.

\begin{example}[label=ex:gnn_social_network]
    Consider a GNN $N$ used to identify bots in a fictitious social network \texttt{NEWSNET} (see Figure~\ref{fig:gnn_social_network}). The network forms a graph where any registered account is a node and the edges represent relationships between accounts. In such a setting, the GNN decides, based on account-specific features, such as the frequency of posts or time of activity, whether an account is a regular human user or a (malicious) bot.

 For the trustworthiness of \texttt{NEWSNET}, it is essential that the GNN $N$  comes with guarantees such as:

    \begin{enumerate}
    \itemsep0em
        \item[A\phantom{'}] ``Every account that spams $100$ messages or more per minute is identified as a bot by $N$.''
        \item[A'] ``If there is a significant, humanly impossible activity of a user within a short amount of time, then $N$ will identify it as a bot.''
        \item[B\phantom{'}] ``Every account whose friends send more than $1000$ messages per minute in total is flagged as part of a bots' network by $N$.''
        \item[B'] ``If an account is a friend with active bots, then $N$ will flag it as a bot by association.''
    \end{enumerate}
    The properties $A$ and $B$ can be seen as safety properties, indicating that $N$ is somewhat robust in its decision making. The properties $A'$ and $B'$ provide reasonable explanations for the decision-making process of $N$.
\end{example}

 \begin{figure}
     \centering
     \tikzset{every picture/.style={line width=0.75pt}} 

\newcommand{\pointedperson}[2]{
  \node[inner sep=0.1mm] (#1#2) at (#1, #2) {\faUser};
}

\newcommand{\person}[2]{
  \node[gray, inner sep=0.1mm] (#1#2) at (#1, #2) {\faUser};
}
\newcommand{\linkpersons}[4]{
 \draw[gray] (#1#2)--(#3#4);
}

\newcommand{\examplesocialetwork}{
    \begin{tikzpicture}[xscale=0.7, yscale=0.4]
        \pointedperson 34
        \person 14
        \person 15
        \person 24
        \person 25
        \linkpersons 3424
        \linkpersons 1524
        \linkpersons 1424
        \linkpersons 3425
    \end{tikzpicture}
}

\begin{tikzpicture}
    \node (input) at (-2.5, 0)  {\examplesocialetwork};
    \node[draw, inner sep=2mm, rounded corners=2mm] (gnn) {GNN $N$};
    \node at (2.5, 0) (output) {bot or user};
    \draw (gnn) edge[-latex] (output);
    \draw (input) edge[-latex] (gnn);
\end{tikzpicture}
     \caption{Setting where a GNN $N$ is used to identify whether an account \faUser\ is a human or a bot in a fictitious social network.}
     \label{fig:gnn_social_network}
 \end{figure}
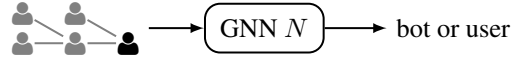

Unfortunately, like most neural network-based models, GNNs exhibit a black-box nature, making them particularly challenging to analyze. The ultimate goal in this context is \emph{formal reasoning}, which involves using sound and complete procedures to interpret the behavior of GNNs, and to certify or falsify specific safety properties of GNNs. Unfortunately, recent works such as \cite{SalzerL23,benedikt2024decidability} suggest that formal reasoning is practically intractable for highly expressive models like GNNs. Consequently, most of the research so far has focused on less rigorous types of reasoning, such as sound but incomplete or probabilistic procedures. For a comprehensive overview of non-formal verification procedures, see \citet{Gunnemann2022}.


\changed{
Several works have explored the expressivity and the verification of GNNs through logical frameworks (\cite{DBLP:conf/iclr/BarceloKM0RS20,benedikt2024decidability,ijcai2024,AhvonenHKL24,DBLP:conf/kr/CucalaG24}).
However, much of this work focuses on
idealized GNNs with real-valued or unbounded integer parameters, whereas practical GNNs are typically \emph{quantized}---their numerical parameters and internal computations are constrained by fixed-size representations. For example, the \textsf{PyG} library
(\cite{Fey/Lenssen/2019}) uses $32$- or $64$-bit floating-point arithmetic, and recent research increasingly targets quantized models with reduced precision (\cite{GholamiQuantized,ZhuLMH0LL023}). In addition, most existing results focus on eventually constant activation functions such as truncated ReLU (\cite{DBLP:conf/iclr/BarceloKM0RS20,ijcai2024}), or provide only partial results for standard ReLU
(\cite{benedikt2024decidability}). Orthogonally, \cite{AhvonenHKL24} investigates the expressivity of \emph{recurrent} GNNs over floats, while \cite{DBLP:conf/kr/CucalaG24} addresses a weaker class of GNNs but considers a broader range of activation functions.
}

The contributions of the paper are as follows.
\begin{itemize}
\item We introduce a logic \thelogic{} capturing a meaningful class of quantized (aggregate-combine) GNNs as well as linear input and output constraints on these GNNs (Section~\ref{sec:logic}).
\item We introduce the linear-constrained validity problem (LVP) that enables to model properties like the ones in Example~\ref{fig:gnn_social_network}. We show how to solve LVP by reducing it to the satisfiability problem of \thelogic{} (\Cref{sec:gnntologic}).
\item We provide a proof system that enables to create graph counterexamples. A prototype of the proof system is implemented in Python. Noticeably, the proof system gives a PSPACE upper bound for a large class of GNNs, with any reasonable activation function (\Cref{sec:tableau}).
\item We also study the PSPACE lower bound and establish that LVP as well as satisfiability of \thelogic{} 
are indeed PSPACE-complete (\Cref{sec:lowerbound}).
\end{itemize}

\section{Fundamentals}
\label{sec:fundamentals}

\newcommand{\setnumbers}{\mathbb K}

\paragraph{Fixed-width Arithmetic} We consider different sets of numbers, denoted by $\setnumbers$, represented in binary using a fixed number of bits. Depending on the underlying (fixed-width) arithmetic, like fixed- or floating-point arithmetic, we denote the respective variants of operations like $+, \cdot, \div$ or relations like $\leq, \geq, =$ by $+_\setnumbers, \leq_\setnumbers, \dots$ and so on. For mathematically rigorous definitions of such arithmetics, see, for example,
\citet{Ercegovac2004,Goldberg91}.
To keep the notation uncluttered, we use $\setnumbers$ to refer to the
set of values as well as the underlying fixed-width arithmetic.
In general, we assume $\setnumbers$ contains $0$, $1$ and $-1$.


\begin{example}
\label{example:fixedpointarithmetic}
$16$-bit fixed-point arithmetic may be used in micro-controllers. For instance, $\setnumbers$ may be the set of numbers, of the form 
    $(-1)^s\frac{k}{10^{4}}$ with four decimal places precision,
    where $s \in \set{0, 1}$ is the sign, $k \in \{0, \dots, 2^{15}-1\}$.
\end{example}

\begin{example}
\label{example:floatingpointarithmetic}
In IEEE754 $32$-bit floating-point arithmetic, $\setnumbers$ exactly contains $0$, $-\infty$, $+\infty$ and numbers that are of the form 
    $(-1)^s 2^e(1+\frac{k}{2^{23}})$
    where $s \in \set{0, 1}$ is the sign, $e \in \set{-128, -125, \dots, 127}$, $k \in \set{0, \dots, 2^{23}-1}$. 
\end{example}

\paragraph{Graphs} We consider directed, vector-labelled graphs $G = (\nodes, \edges, \ell)$ where $\nodes$ is a finite set of nodes, or vertices,
$\edges$ is a set of directed edges and $\ell: V \rightarrow \mathbb{R}^m$.
We denote the class of all such graphs by $\graphs$.
We call a pair $(G,v)$ where $v$ is a node of $G$ a \emph{pointed graph}. We denote the class of all pointed graphs by $\pgraphs$.
If we restrict the classes $\graphs$ and $\pgraphs$ to vector labels with values from some fixed-width arithmetic $\setnumbers$, we denote this by $\graphs[\setnumbers]$ and $\pgraphs[\setnumbers]$.

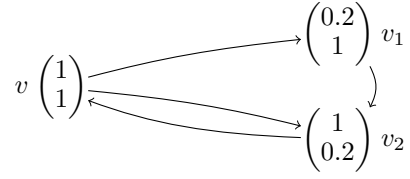
\begin{figure}
    \centering
    	\begin{tikzpicture}[xscale=2,yscale=0.7]
		\draw[fill=white, opacity=0.7,rounded corners = 5mm, draw=none] (-2.5, 1.5) rectangle (0.6, -1.5);
		\node[inner sep=0mm, rounded corners=2mm] (0) at (-2, 0) {\labellednode{v}{1}{1}};
		\node[inner sep=0mm] (1) at (0, 1) {\labellednodeinv{v_1}{0.2}{1}};
		\node[inner sep=0mm] (2) at (0, -1) {\labellednodeinv{v_2}{1}{0.2}};
		\draw[->] (0) edge [bend left=10] (1);
		\draw[->] (0) edge [bend left=10] (2);
		\draw[->] (1) edge [bend left=10] (2);
		\draw[->] (2) edge [bend left=20] (0);
	\end{tikzpicture}
    \caption{Example of a vector-labelled graph.}
    \label{fig:graph}
\end{figure}
\begin{example}
    Figure~\ref{fig:graph} shows a vector-labelled graph with three nodes. Here $\ell : V \rightarrow \ensR^2$ and $\ell(v) = \begin{pmatrix}
        1\\1
    \end{pmatrix}$ for instance.
\end{example}

\paragraph{Graph Neural Networks} 
A \emph{Graph Neural Network (GNN)} $N$ is a tuple $(\layer_1, \dotsc, \layer_m, \layer_{\mathit{out}})$. Each layer $\layer_i$ computes 
$\layer_i(\boldsymbol{x}, M) = \comb_i(\boldsymbol{x}, \agg_i(M))$ where $\boldsymbol{x}$ is a real-valued vector, $M$ is a 
multiset of vectors, $\agg$ is called an \emph{aggregation function}, mapping a multiset of vectors onto a single vector, 
and $\comb$ is called a \emph{combination function}. 
Unless explicitly stated otherwise, $\agg$ denotes the sum.
Layer $\layer_{\mathit{out}}$ computes a function from vectors to vectors.
The GNN $N$ computes a function $\pgraphs \rightarrow \mathbb{R}^n$, where $n$ is the \emph{output dimension
of $N$}, in the following way. Let $(G, v)$ where $G=(V,E,\ell)$ be a pointed graph. 
First, $N$ computes a state $\boldsymbol{x}^m_u$
for all $u \in V$ in a layer-wise fashion by $\bs{x}^i_u = \comb_i(\boldsymbol{x}^{i-1}_u, \agg_i(\multset{\bs{x}^{i-1}_{u'} \mid (u,u') \in \edges}))$ 
and $\bs{x}^0_u = \ell(u)$. The overall output $N(G,v)$ is given by $\layer_{\mathit{out}}(\bs{x}^m_{v})$.
If we restrict the computation of $N$ to some fixed-width arithmetic, it computes a function 
$\pgraphs[\setnumbers] \rightarrow \setnumbers^n$ and we assume that all numerical parameters of $N$ as well as internal
computations are carried out in $\setnumbers$. We call $N$ \emph{quantized (to $\setnumbers$)} in this case.

\begin{example}\label{ex:basic}
    Consider a GNN $N = (\layer_1, \layer_2, \layer_{out})$ with $\layer_i(\boldsymbol{x}, M) = \comb(\boldsymbol{x}, \agg(M))$
    where $\agg$ is the sum and with the same combination function $\comb: \mathbb{R}^{2\cdot2} \rightarrow \mathbb{R}^2$ in both layers where
    $\comb((x, x'),(\aggvar, \aggvar'))$ is given
    by
    \begin{displaymath}
    \left(\begin{matrix}
		\activationfunction(x + 2x' - 3\aggvar + 4\aggvar' + 5) \\
		\activationfunction(6x + 7x' + 8\aggvar - 9\aggvar' + 10) \\
	\end{matrix}\right)
 \end{displaymath}
 for some activation function $\activationfunction$, such as ReLU. The output layer $\layer_{out}$ 
 computes $\mathbb{R}^2 \rightarrow \mathbb{R}$ and is
 given by $\activationfunction(x_1 + x_2 - 2)$.
\end{example}

\section{Boolean Combinations Over Modal Expressions}
\label{sec:logic}

To reason about broad classes of GNNs and overcome the limitations of existing logical characterizations, we propose the logic $\thelogic{}$, which enables reasoning about general but quantized GNNs with arbitrary fixed-width arithmetic and activation functions. Existing characterizations, such as graded modal logic (\cite{DBLP:conf/iclr/BarceloKM0RS20}), modal logic on linear inequalities over counting (\cite{ijcai2024}), or fragments of Presburger logic (\cite{benedikt2024decidability}), typically relate to GNNs with eventually constant activation functions, like truncated ReLU $x \mapsto \max(0, \min(1,x))$. In contrast, $\thelogic{}$ extends beyond these constraints using the sole assumption of a quantized setting.

\subsection{Syntax}
Let $\featuresset$ be a finite set of features and $\setnumbers$ be some finite-width arithmetic. We consider a set of \emph{expressions} defined by the following grammar:
$$\expression ::= c \mid x_i \mid \activationfunction (\expression) \mid \agreggationfunction (\expression) \mid 
\expression + \expression \mid c \times \expression$$
where $c \in \setnumbers$, $x_i \in \featuresset$.
A \emph{formula} of the logic $\thelogic{}$ is then a Boolean formula where the atomic formulas are of the form $\expression \geq k$ for some $k \in \setnumbers$. 
When there is no ambiguity and to lighten the notation, we sometimes use more standard expressions, like $-x$ instead of $-1 \times x$, 
or $c\expression$ instead of $c \times \expression$, etc.
We use $\expression = k$  as an abbreviation for
$(\expression \geq k) \land (-1 \times \expression + k \geq 0)$.
Formulas are represented as DAGs.
For example, in 
$(\expression \geq k) \land (-1 \times \expression + k \geq 0)$
the expressions $\expression$ and $k$ are \emph{not} duplicated as shown in Figure~\ref{figure:dag}.
The symbol $\top$ denotes the tautology 
(any universally true formula, e.g. $x_1 - x_1 \geq 0$).
\begin{figure}
    \centering
    \begin{tikzpicture}[xscale=1.3,yscale=0.5, 
    every node/.style={inner sep=0.5mm}
    ]
        \node (E) at (1, 0) {$\expression$};
        \node (-1) at (1, -2) {$-1$};
        \node (k) at (1, -1) {$k$};
        \node (0) at (1, -3) {$0$};        
        \node (plus) at (-1.5, -1.25) {$+$};
        \node (times) at (-0.25, -2) {$\times$};
        \node (geq1) at (-2, 0) {$\geq$}; 
        \node (geq2) at (-2, -2.5) {$\geq$}; 
        \node (and) at (-3, -1) {$\land$};

        \draw[-latex] (geq1) edge[bend left=10] node [above] {$\cdot$} (E);
        \draw[-latex] (geq1) edge[bend right=10] (k);

        \draw[-latex] (geq2) edge[bend left=10] node [left] {$\cdot$} (plus);
        \draw[-latex] (geq2) edge[bend right=10] (0);

        \draw[-latex] (plus) edge[bend left=5] (k);
        \draw[-latex] (plus) edge[bend right=5] (times);
        \draw[-latex] (times) edge[bend right=5] (E);
        \draw[-latex] (times) edge[bend right=20] (-1);
          
        \draw[-latex] (and) -- (geq1);
        \draw[-latex] (and) -- (geq2);
    \end{tikzpicture}
  
    \caption{DAG representation for 
    $(\expression \geq k) \land (-1 {\times} \expression + k \geq 0)$.}
    \label{figure:dag}
\end{figure}

\subsection{Semantics}

\newcommand{\semone}[1]{[\![#1]\!]}
\newcommand{\sem}[2]{[\![#1]\!]_{#2}}

Recall that $\setnumbers$ is a finite set of numbers, equipped with an addition $+_\setnumbers$, multiplication $\times_\setnumbers$ operation, and comparison operations $\leq_\setnumbers$, $<_\setnumbers$.
We fix $\semone{\activationfunction}$ to be an activation function. For instance, $\semone{\activationfunction} = ReLU$, with $ReLU(x) = x$ if $x \geq_\setnumbers 0$ and $0$ otherwise. 
Formulas are evaluated over the class of pointed, labelled graphs $G = (\nodes, \edges, \ell)$ where $\ell : \nodes \rightarrow \setnumbers^I$, defined as follows.
The semantics $\sem{\expression} u$ of an expression $\expression$ with respect to a vertex $u\in V$ is inductively defined on $\expression$: 
\begin{center}
\begin{minipage}{40mm}
\begin{align*}
\sem {c} u & = c \\
    \sem {x_i} u & = \ell(u)_i \\
      \sem{\expression + \expression'} u & = \sem \expression u +_{\setnumbers} \sem {\expression'} u 
  \end{align*}
\end{minipage}
\hfill
\begin{minipage}{35mm}
\begin{align*}
  \sem{c \times \expression} u & = c \times_{\setnumbers} \sem{\expression} u \\
      \sem{\activationfunction (\expression)} u & = \semone{\activationfunction}(\sem \expression u)  \\
    \sem{\agreggationfunction (\expression)} u & = \Sigma_{v \mid u \edges v}\sem {\expression} v
\end{align*}
\end{minipage}
\end{center}

The semantics $\sem{x_i}{u}$ is the value of $i$-th feature in vertex~$u$. 
Note that $\agreggationfunction(\expression)$
is motivated by the message-aggregation process in a GNN and is the `modal' part of $\thelogic{}$. Its semantics $\sem{\agreggationfunction(\expression)}{u}$ consists in summing up the values/semantics of $\expression$ in the successors $v$ of $u$.
We write $G, u \models \phi$ to say that the formula $\phi$ is true in the pointed graph $(G,u)$, and we say that $(G,u)$ is a \emph{model} of $\phi$. Given $\degree \in \mathbb{N} \union \set{+\infty}$, a formula $\phi$ is \emph{$\delta$-satisfiable} 
(simply satisfiable if $\delta = +\infty$)
if it has a model in which each vertex is of out-degree (number of successors) at most $\delta$. 

\begin{example}
The formula $x_1 + \activationfunction(x_2) \geq 0$ is true in the pointed graphs $(G, u)$ where the sum of the first feature at $u$ plus the result of the application of the activation function on the second feature at $u$ is positive (i.e., $(\ell(u))_1 + \semone{\activationfunction}(\ell(u))_2 \geq_{\setnumbers} 0$).
The formula $\agreggationfunction(1) = 4$ is true in the pointed graphs $(G,u)$ with $u$ having exactly $4$ successors. The formula $\agreggationfunction(3) = 10$ is unsatisfiable.
\end{example}

\section{Enabling Formal Reasoning About GNNs Using $\thelogic{}$}
\label{sec:gnntologic}

Let us now formally define \emph{linear-constrained validity problem (LVP)} for GNNs.

\begin{definition}
\label{definition:lvp}
The \emph{linear-constrained validity problem (LVP)} 
is defined as follows:
    \begin{itemize}
        \item given a quantized GNN $N$ with input dimensionality $m$ and output dimensionality $n$, two systems of linear inequalities $L_{\text{in}}$ over variables $x_1, \dotsc, x_m$ and 
$L_{\text{out}}$ over variables $y_1, \dotsc, y_n$, and a bound $\delta \in \mathbb{N} \cup \{+\infty\}$,
\item decide whether 
$(N, L_{\text{in}}, L_{\text{out}}, \delta)$ 
is \emph{valid}, that is all pointed graphs $(G,v)$ 
where the out-degree of each vertex is less or equal than $\delta$, $v$ has $m$ features, and $v$ satisfies system $L_{\text{in}}$ imply that $N(G,v)$ satisfies system~$L_{\text{out}}$. 
    \end{itemize}
\end{definition} 
We write 
$(N, L_{\text{in}}, L_{\text{out}})$ instead of 
$(N, L_{\text{in}}, L_{\text{out}}, +\infty)$.

\begin{example}[continues=ex:gnn_social_network]
Recall the setting of Example~\ref{ex:gnn_social_network} where we considered the safety property, “Every account that spams $100$ messages or more per minute is identified as a bot by $N$.” 

This property is expressed by the LVP instance $(N, L_{\text{in}}, L_{\text{out}})$ with
$L_{\text{in}} := x_1 \geq 100$ and $L_{\text{out}} := y_1 \geq 0.6$
where 
$x_1$ 
is the feature corresponding to the number of messages sent by a user, and 
$y_1$ is the output of $N$ corresponding to its classification as a bot 
($y_1 \geq 0.6$) or a regular user ($y_1 < 0.6$). 
Note that in this simple example $L_{\text{in}} $ and $L_{\text{out}}$ are systems consisting of a single inequality.
\end{example}

We show that LVP over quantized 
GNNs, where aggregation is given by summation, and combination functions and output functions are realized by classical feedforward neural networks (FNNs) using some kind of activation $\activationfunction$, is reducible to the satisfiability problem of \thelogic{}.

We denote the size of the GNN $N$ by $|N|$, with $|N| \in \mathcal{O}(|N_{\comb_1}| + \dotsb + |N_{\comb_m}| + |N_{\mathit{out}}|)$ where 
$m$ is the depth of $N$ and $|N_{\comb_i}|$ and 
$|N_{\mathit{out}}|$ are the sizes of the FNN used in $N$. The size of an FNN is simply the sum of the sizes of all its numerical parameters, namely weights and biases.

\begin{theorem}
\label{th:reduction}
    Let $I = (N, L_{\text{in}}, L_{\text{out}}, \degree)$ be an LVP instance. There is a \thelogic{} formula $\varphi_I$ such that $I$ is valid 
    if and only if $\varphi_I$ is not $\degree$-satisfiable. Furthermore, $\varphi_I$ can be computed 
    from $N, L_{\text{in}}, L_{\text{out}}$
    in polynomial time with respect to $|N| + |L_{\text{in}}| + |L_{\text{out}}|$.
\end{theorem}

\begin{proof}
    Let $I = (N, L_{\text{in}}, L_{\text{out}})$ be an LVP instance where the GNN $N$ has $l$ layers $\layer_i$, each having a combination function $\comb_i$
    represented by an FNN $N_{\comb_i}$ as well as aggregation function $\agg_i$ represented by the sum $\sum$ and output function represented by FNN $N_{\mathit{out}}$, $L_{\text{in}}$ consists of linear inequalities
    $\varphi_1, \dotsc, \varphi_{k_1}$ working over variables $x_1, \dotsc, x_{n_1}$ and $L_{\text{out}}$ consists of linear inequalities
    $\psi_1, \dotsc, \psi_{k_2}$ working over variables $y_1, \dotsc, y_{n_2}$.
    Formula $\varphi_I$ is the conjunction $\varphi_1 \land \dotsb \land \varphi_{k_1}\land \varphi_N \land (\neg\psi_1 \lor \dotsb \lor \neg\psi_{k_2})$ where $\varphi_N$ is constructed as follows. 
    Formula $\varphi_N$ is a conjunction
    $\expression_{N,y_1} - y_1 = 0 \land \dotsb \land \expression_{N,y_n} - y_{n_2} = 0$ where $\expression_{N, y_i}$ captures the computation of $N$ corresponding to its $i$-th
    output as follows.

    Consider layer $\layer_1$ and assume that $N_{\comb_1}$ has input dimension $2i_1$ and output dimension $o_1$. We define the sub-expression $\expression_{\layer_1, j}= \expression_{N_{\comb_1},j}(z_1, \dotsc, z_{i_1}, \agg(z_1), \dotsc, \agg(z_{i_1}))$ for $j \leq o_1$ where $\expression_{N_{\comb_1},j}$ is the straightforward unfolding of the function of $N_{\comb_1}$ corresponding to its $j$-th output using $c, c\expression, +$ and $\activationfunction(\expression)$ operators on the inputs 
    $z_1, \dotsc, \agg(z_{i_1})$. Next, consider the
    layers $\layer_h$ for $2 \leq h \leq l$. We assume that $N_{\comb_h}$ has input dimension $2o_{h-1}$ and output dimension $o_h$.
    We define the sub-expression $\expression_{\layer_h, j}$ equal to
    $\expression_{N_{\comb_h},j}(\expression_{\layer_{h-i},1}, \dotsc, \expression_{\layer_{h-1},o_{h-1}}, \agg(\expression_{\layer_{h-1},1}),$ $\dotsc$, $\agg(\expression_{\layer_{h-1},o_{h-1}}))$ for each $j \leq o_h$. Finally, we define $\expression_{N,x_j} = \expression_{N_{\mathit{out}, j}}(\expression_{\layer_m, 1}, \dotsc, \expression_{\layer_m, o_m})$ where $\expression_{N_{\mathit{out}, j}}$ is again the unfolding of the function of $N_{\mathit{out}}$ 
    corresponding to its $j$-th output.

    The correctness follows 
    from the construction above, as we simply unfold the computation of the GNN~$N$ using operators available
    in \thelogic{}. The polynomial bound is due to the fact that we represent formulas as DAGs and, thus, avoiding any duplication of subformulas.
\end{proof}

\begin{example}
Consider the simple GNN $N$ of Example~\ref{ex:basic}.
The formula equivalent to $N$ in the sense of Theorem~\ref{th:reduction} is $\varphi_N := \activationfunction(\expression_{2}+ \expression_{2}') - y = 0$ with
\begin{align*}
    \expression_{i+1} & = \activationfunction(\expression_{i} + 2\expression_{i}' - 3\agg (\expression_{i}) + 4\agg (\expression_{i}') + 5), \\
    \expression_{i+1}' & = \activationfunction(6\expression_{i} + 7\expression_{i}' +8\agg( \expression_{i}) -9\agg (\expression_{i}') + 10),
\end{align*}
$\expression_0 = x$ and $\expression'_0 = x'$. Assuming that $L_{\text{in}} := x+x'\geq 0$ and
$L_{\text{out}} := y \geq 0.6$, the overall \thelogic{} formula would be
$\varphi_I := x+x'\geq 0 \land \varphi_N \land \lnot (y \geq 0.6)$.
\end{example}

\begin{remark}
The translation in Theorem~\ref{th:reduction} can be straightforwardly extended beyond the capabilities of LVP. For instance, we can generalize to instances of the form $(N, \varphi, \psi)$, where~$N$ is a GNN and $\varphi, \psi$ are \thelogic{} formulas.
\end{remark}

Note that if an LVP instance $I$ is not valid, then  $\varphi_I$ is satisfiable. A model of $\varphi_I$ offers a counterexample for $I$.

\section{Satisfiability of $\thelogic{}$ Formulas}
\label{sec:tableau}

In this section, we tackle the following satisfiability problem, in which $\degree$ is a bound on the out-degrees of the pointed graph and $n$ is the number of bits to represent an element of $\setnumbers$. For instance, for $n = 32$, $\setnumbers$ could be 
\Cref{example:floatingpointarithmetic}. 
We write $\setnumbers_n$ to emphasize that the numbers are encoded on $n$ bits. In the sequel, we fix a sequence $\setnumbers_0, \setnumbers_1, \setnumbers_2 \dots$ with the following reasonable poly-space assumption for each $\setnumbers_n$:
\begin{enumerate}
    \item there is a uniform algorithm, taking $n$, $k_1$, $k_2$ and $k$ in $\setnumbers_n$ as an input and checking whether $k_1{+}k_2{=}k$ in poly-space in $n$ (same for whether $k_1{=}k_2$, $k{\times} k_1{=}k_2$).
    \item given $n$ written in unary, $k, k' \in \setnumbers_n$ written in binary, deciding whether
$\semone{\activationfunction}(k') = k$
can be performed in polynomial space in $n$. This assumption is weak and is true for a large class of activation functions $\semone{\activationfunction}$ (for instance $ReLU$, $truncatedReLU$ for fixed-point or floating-point arithmetics on $n$ bits).
\end{enumerate}

\begin{definition}[satisfiability problem]
\label{definition:satisfiabilityproblem}
The satisfiability problem of \thelogic{} is:
    \begin{itemize}
        \item input: an integer $n$ written in unary, an integer $\degree$ (or $\degree = +\infty$) in unary or in binary, a formula $\Phi$ in \thelogic{}
        \item question: does there exist a pointed graph with out-degree at most $\degree$ with features in a set $\setnumbers_n$ that satisfies~$\phi$?
    \end{itemize}
\end{definition}

Note that $n$ is written unary is reasonable since $n$ is precisely of the same order of magnitude as the size of the representation of an element in $\setnumbers_n$, and we have to store the bits of numbers in $\setnumbers_n$ fully in memory anyway.

Bounding the out-degree $\delta$ is justified in applications such as chemistry.
Generally, it is sufficient to represent $\degree$ in unary since the bound is often small. For instance, a carbon atom has $4$ neighbors. Having $\degree$ in unary means also that we can keep all the successors/neighbors of a given vertex in memory at once.
A graph can be exponential in $\degree$; a tree may contain $\degree^{depth}$ nodes, making the search space exponential in the input size when the depth is unbounded. Using a binary representation for $\delta$ is also interesting since the search space is already exponential with a fixed depth. For instance, we could search for a pointed graph which is a tree of depth $1$ where the root has an exponential number of children.

\subsection{Tableau Method}

\begin{figure*}
\small
    \centering

\begin{minipage}{0.45\linewidth}
     \begin{minipage}{0.45\textwidth}
 \begin{prooftree}
     \AxiomC{$(w~\phi \lor \psi)$}
     \rulelabel{\lor}
     \UnaryInfC{$(w~\phi) \mid (w~\psi)$}
 \end{prooftree}

 \begin{prooftree}
     \AxiomC{$(w~\lnot(\phi \lor \psi))$}
     \rulelabel{\lnot\lor}
     \UnaryInfC{$(w~\lnot \phi), (w~\lnot \psi)$}
 \end{prooftree}

 \begin{prooftree}
     \AxiomC{$(w~\phi \land \psi)$}
     \rulelabel{\land}
     \UnaryInfC{$(w~\phi), (w~\psi)$}
 \end{prooftree}
 \end{minipage}
 ~
 \begin{minipage}{0.45\textwidth}
 \begin{prooftree}
     \AxiomC{$(w~\lnot (\phi \land \psi))$}
     \rulelabel{\lnot\land}
     \UnaryInfC{$(w~\lnot \phi) \mid (w~\lnot \psi)$}
 \end{prooftree}

 \begin{prooftree}
     \AxiomC{$(w~\lnot \lnot \phi)$}
     \rulelabel{\lnot\lnot}
     \UnaryInfC{$(w~ \phi)$}
 \end{prooftree}
 \end{minipage}

\begin{prooftree}
    \AxiomC{$ (w~\lnot(\expression \geq k))$}
    \rulelabel{\lnot\geq}
    \UnaryInfC{$(w~\expression=k')$ for some $k' \in \setnumbers_n$ with $k' <_{\setnumbers_n} k$}
\end{prooftree}


\begin{prooftree}
\AxiomC{$(w~\expression_1 + \expression_2 = k)$}
\rulelabel{+}
\UnaryInfC{\begin{minipage}{0.8\textwidth}
$(w~\expression_1 = k_1), (w~\expression_2 = k_2)$ for some $k_1, k_2 \in \setnumbers_n$,  with $k_1 +_{\setnumbers_n} k_2 = k$
\end{minipage}
}
\end{prooftree}

\begin{prooftree}
    \AxiomC{$(w~\phi)$ and no term $(w~degree=...)$}
\rulelabel{degree}
\UnaryInfC{$(w~degree=\delta')$ for some $\delta'\leq \delta$}
\end{prooftree}
\end{minipage}
\begin{minipage}{0.45\linewidth}

\begin{prooftree}
\AxiomC{$(w~c = k)$ if $c \neq k$}
\rulelabel{clash_c}
\UnaryInfC{fail}
\end{prooftree}

\begin{prooftree}
\AxiomC{$(w~x_i = k), (w~x_i = k')$ if $k \neq k'$}
\rulelabel{clash_=}
\UnaryInfC{fail}
\end{prooftree}

\begin{prooftree}
\AxiomC{$(w~\activationfunction(\expression) = k)$}
\rulelabel{\activationfunction}
\UnaryInfC{$(w~\expression = k')$ for some $k' \in \setnumbers_n$ with $\semone{\activationfunction}(k') = k$}
\end{prooftree}

\begin{prooftree}
\AxiomC{$(w~\expression \geq k)$}
\rulelabel{\geq}
    \UnaryInfC{$(w~\expression=k')$ for some $k' \in \setnumbers_n$ with $k' \geq_{\setnumbers_n} k$}
\end{prooftree}


\begin{prooftree}
\AxiomC{$(w~c\expression = k)$}
\rulelabel{\times}
\UnaryInfC{$(w~\expression = k')$ for some $k' \in \setnumbers_n$ with $c \times_{\setnumbers_n} k' = k$}
\end{prooftree}

\begin{prooftree}
\AxiomC{$(w~\agg(\expression) = k)$}
\AxiomC{$(w~degree=\delta')$}
\rulelabel{\agg}
\BinaryInfC{
\begin{minipage}{0.8\textwidth}
$(w1~\expression = k_{1}), \dotsc, (w \degree'~\expression = k_{{\degree'}})$ for some $(k_u)_{u=1..\degree'}$, with $k_{1} +_{\setnumbers_n} 
 \dotsb +_{\setnumbers_n} k_{\degree'} = k$
\end{minipage}
}
\end{prooftree}

\end{minipage}

    \caption{Tableau rules for the satisfiability problem in \thelogic{} when $\delta$ is written in unary.
}
    \label{figure:tableaurules}
\end{figure*}

We propose a tableau method (\cite{DBLP:books/el/RV01/Hahnle01,DBLP:books/daglib/0032750}) to reason in our logic, i.e. to check whether a formula $\Phi$ is $\delta$-satisfiable. For presentation purposes, we focus on the case where the degree bound $\degree$ is finite and given in unary. The cases where $\degree$ is written in binary or $\delta = +\infty$ are more involved, and are discussed in \citesuplementarymaterial.

The idea of tableau methods is to construct a model satisfying an initial formula by propagating constraints and creating new vertices when needed.
We consider terms of the form $(w~\phi)$ where $w$ is a word (that intuitively denotes a vertex in the graph in construction) and $\phi$ is a formula. The intuitive meaning of $(w~\phi)$ is that $\phi$ should be true in the vertex $w$.

We start the tableau method with the following term $(\epsilon ~ \Phi)$, saying that $\Phi$ should be true in the initial vertex, denoted by the empty word $\epsilon$.
Figure~\ref{figure:tableaurules} shows our tableau rules. 

The rules for $(\lnot\geq), (\geq), (+), (\activationfunction), (\times)$ exploit the fact that the set $\setnumbers_n$ is finite. For instance, rule $(\geq)$ non-deterministically chooses a value for $\expression$ that is greater than $k$ in $\setnumbers_n$.
The rule $(+)$ on the term $(w~\expression_1+\expression_2=k)$ (whose meaning is that the value of $\expression_1+\expression_2$ should be equal to $k$) non-deterministically chooses values $k_1$ and $k_2$ for $\expression_1$ and $\expression_2$ respectively so that $k_1 + k_2 = k$.
The rule $(clash_c)$ stops the current execution when the constant $c$ is different from the value $k$ to which it is constrained to be equal. Similarly, the rule $(clash_=)$ clashes when the feature $x_i$ at the vertex denoted by $w$ is supposed to be equal to $k$ and $k'$ with $k\neq k'$. The rules $(\activationfunction)$ and $(\times)$ work similarly to the rule $(+)$ but for the activation function and for multiplication by a scalar. 
The rule $(degree)$ non-deterministically guesses the out-degree of the vertex specified by~$w$.

The rule $(agg)$ is `modal' and creates new successors $w1, \dots, w\degree'$ where $w1$ is the word $w$ on which we concatenated the letter $1$, and so on. When the aggregation for the expression $\expression$ is supposed to be equal to $k$ and the out-degree is supposed to be $\delta'$, we create $\degree'$ successors and impose that the sum of the values of $\expression$ taken in these successors is equal to $k$ (to this aim, we non-deterministically choose values $k_1, \dots, k_{\degree'}$ for $\expression$ in these successors).

 The proposed tableau method is a non-deterministic procedure that applies as long as possible the rules given in Figure~\ref{figure:tableaurules} from the initial set $\set{(\epsilon~\Phi)}$.
 The method either fails or stops successfully if no more rules are applicable.

%
    
 %


A prototype implementation\footnote{\url{https://github.com/francoisschwarzentruber/ijcai2025-verifquantgnn}}
of the proof system is presented in 
 \citesuplementarymaterial.

\subsection{Upper Bound}

\begin{theorem}
\label{th:satlogPSPACE}
The satisfiability problem of $\thelogic$ is in PSPACE, whether $\delta$ is given in unary, binary or is infinite.
\end{theorem}

\begin{proof}
The proof is given for $\delta$ in unary.  For the proof when~$\delta$ in binary or infinite, we can adapt the proof system, see \citesuplementarymaterial.
   We have to prove that $\Phi$ is satisfiable iff there is an accepting execution of the tableau method.
    
    The \fbox{$\Rightarrow$} direction is proven by taking a pointed graph $G, v$ satisfying $\Phi$ and making the non-deterministic choices by letting oneself be guided by the evaluation of formulas in $G$. 
    To do that, we start the execution of the tableau method with $(\epsilon~~ \phi)$ and we set $f(\epsilon) = v$. We then show that we can always continue the execution maintaining the invariant:
    the partial mapping $f$ from the labels in the tableau to $V$ is such that if $(w~\phi)$ appears in the tableau then $G, f(w) \models \phi$; and if $(w~degree=\delta')$ appears then $f(w)$ is of out-degree $\delta'$.
    
    More precisely, we make the non-deterministic choices so that we can build a partial mapping $f$ from the labels in the tableau to $V$, starting from $f(\epsilon) = v$ so that if $(w~\phi)$ appears in the tableau then $G, f(w) \models \phi$; and if $(w~degree=\delta')$ appears then $f(w)$ is of out-degree $\delta'$. 
    
    Let us explain the principle with the $(agg)$ rule. Suppose $(w~\agg(\expression) = k)$ and $(w~degree=\delta')$ appears in the tableau. By the invariant, it means that $G, f(w) \models \agg(\expression) = k$ and that $f(w)$ is of out-degree $\delta'$. But then $f(w)$ has $\delta'$ successors: $u_1, \dots, u_{\delta'}$. As $G, f(w) \models \agg(\expression) = k$, there exists $k_{1}, \dots, k_{{\degree'}}$ with $k_{1} +_{\setnumbers_n} 
 \dotsb +_{\setnumbers_n} k_{\degree'} = k$ with $G, u_i \models \expression = k_i$ for all $i=1..\delta'$.
We then set $f(w1) := u_1, \dots, f(w\delta') = u_{\delta'}$. We also apply the rule $(agg)$ by choosing $(w1~\expression = k_{1}), \dotsc, (w \degree'~\expression = k_{{\degree'}})$. The invariant remains true.

By the invariant, the premises of rules $(clash_c)$ or $(clash_=)$ are never in the tableau. For instance, if some $(w~~c=k)$ with $c \neq k$ appears, we would $G, w \models c=k$ which is a contradiction.

    The \fbox{$\Leftarrow$} direction is proven by constructing a pointed graph $G, v$ from the choices made in the accepting execution. The graph $G = (V, E)$ is defined with $V$ the set of labels $w$ and $E$ is the set of edges of the form $(w, wu)$ with $w$ a label and $u \in \set{1, \dots, \delta'}$ with $(w~~degree=\delta')$ appearing in the tableau. The label is $\ell(w) = (v_1, \dots, v_m)$ with $v_i$ such that $(w~x_i=v_i)$ appears in the tableau, or $v_i$ can be any number in $\setnumbers_n$ if no term $(w~x_i=...)$ appears.

    We then prove that $G, \epsilon \models \phi$. To do that, we prove that if $(w~~\psi)$ appears in the final tableau, then $G, w \models \psi$ by induction on $\psi$.
    We show it for $\psi =: \agg(\expression) = k$.
    Suppose that $(w~~\agg(\expression) = k)$ appears in the final tableau. Then it means that the rule $(\agg)$ has been applied on it. So $(w1~\expression = k_{1}), \dotsc, (w \degree'~\expression = k_{{\degree'}})$ for some $(k_u)_{u=1..\degree'}$, with $k_{1} +_{\setnumbers_n} 
 \dotsb +_{\setnumbers_n} k_{\degree'} = k$ appear in the final tableau. By induction, it means that $G, wi \models \expression = k_{i}$. So $G, w \models \agg(\expression) = k$.
Finally, as $(\epsilon~~\phi)$ is in the tableau, the induction hypothesis gives $G, \epsilon \models \phi$.

It can be implemented in poly-space, in the same spirit as the tableau method for modal logic K (\cite{DBLP:journals/ai/HalpernM92}). The idea is to backtrack on rules $(\agg)$: first we try to add all the $(w1~\expression = k_{1})$ for all applications of $(\agg)$, then try to add the $(w2~\expression = k_{2})$ and so on. The algorithm is then a depth-first-search on the tree relying on the labels $w$ (see \Cref{figure:searchtree}). We conclude by Savitch's theorem (\cite{DBLP:journals/jcss/Savitch70}) to obtain a deterministic algorithm.
\end{proof}

\begin{figure}
	\centering
		\small
	\begin{tikzpicture}[
		every node/.style = {inner sep=0.5mm},
		level distance = 6mm,
		edge from parent/.style = {draw, -latex},
		sibling distance = 2cm, 
		level 2/.style = {sibling distance=1cm}, 
		]
		
		\node[fill=red!30, draw=red] (r) {$\epsilon$}
		child {node {1}
			child {node {11}}
			child {node {12}}
		}
		child {node[fill=red!30, draw=red] (2) {2}
			child {node[fill=red!30, draw=red] (21) {21}}
			child {node {22}}
		}
		child {node {3}
			child {node {31}}
			child {node {32}}
		};
	
	\path[draw=red, thick, -latex] (r) -- (2);
	\path[draw=red, thick, -latex] (2) -- (21);
		
	\end{tikzpicture}
	\caption{Depth-first search tree where each node contains terms of the form $(w~~...)$ for a given word $w$ among $\epsilon, 1, 2, 3, 11, 12, 21, 22, 31, 32$. Only a single branch is in memory.
 }
  \label{figure:searchtree}
 
\end{figure}
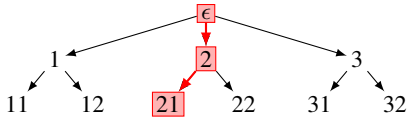

Now, Theorem~\ref{th:reduction} and \ref{th:satlogPSPACE} imply an upper bound for LVP. 
Actually, we consider a quantized variant of LVP. Extending Definition~\ref{definition:lvp}, that problem also takes $n$ in unary and $\degree$ in unary or in binary (as in Definition~\ref{definition:satisfiabilityproblem}).

\begin{corollary}
\label{cor:LVPPSPACE}
    LVP is in PSPACE.
\end{corollary}

Note that when $\degree$ is in unary, the satisfiability problem is in NP when the number of nested $\agg$ operators in the input formula is bounded by a fixed integer $i$. Indeed, the certificate is a pointed labelled tree of depth $i$ of out-degree at most $\delta$ (it contains $O(\delta^i)$ vertices). In consequence, LVP when we bound the number of layers by~$i$ is in coNP.

\subsection{Usage Example}

\label{sec:complete-example}
We assume that $\setnumbers$ is signed $32$-bit fixed-point numbers, with four decimal places precision, similar to Example~\ref{example:fixedpointarithmetic}.
The number $\frac{1}{125}$ is stored as (the binary encoding of the decimal number) $80$, $\frac{1}{1000}$ is stored as $10$, $0.9$ is stored as $9000$, $1000$ is stored as $10000000$, $250$ is stored as $2500000$, etc.
%

Elaborating on \Cref{ex:gnn_social_network}, consider the GNN $N = (\layer_1, \layer_{out})$ where $\layer_1 = (\agg_1, \comb_1)$, $\agg_1 = \sum$, $\comb_1((x_1), (\aggvar_1)) = \begin{pmatrix}
    ReLU(+\frac{1}{125}x_1 + 0\aggvar_1 + 0) \\
    Id(+0x_1 + \frac{1}{1000}\aggvar_1+ 0)
    \end{pmatrix}
    $.
So $N$ has only one layer, with input dimension $1$ (the value of $x_1$ that labels a node indicates the number of messages per minute sent by the individual), and output dimension~$2$. The output layer is the identity.
Consider again the LVP $I = (N, L_{\text{in}}, L_{\text{out}})$ where $L_{\text{in}} = x_1 \geq 100$ and $L_{\text{out}} = y_1 \geq 0.6$.

\paragraph{$\thelogic{}$ formula.}

As per \Cref{th:reduction} and its proof,
the LVP $I = (N, L_{\text{in}}, L_{\text{out}})$ is valid iff $\phi_I = \phi_{\text{in}} \land \phi_N \land \lnot \psi_{\text{out}}$ is not satisfiable, where:
\begin{itemize}
    \item $\phi_{\text{in}} := x_1 \geq 100$
    \item $\phi_N := ReLU(\frac{1}{125} x_1) - y_1 = 0 \land \frac{1}{1000}\agg(x_1) - y_2 = 0$
    \item $\psi_{\text{out}} := y_1 \geq 0.6$
\end{itemize}

\paragraph{Automated reasoning.}

We could show that
$(\agg(x_1) \geq 1000) 
\land \phi_N \land
\lnot (y_2 \geq 1.0)
$ is not satisfiable, establishing that $N$ comes with the guarantee of property $B$. Here we focus on property $A$.
If the pointed graph $(G,v)$ satisfies $L_{\text{in}}$, then $N(G,v)$ will satisfy $L_{\text{out}}$.
We could use the tableau method to automatically show that $\phi_I$ is not satisfiable and therefore that $I$ is valid. Unfortunately, a complete formal proof cannot be reasonably presented here or even done by hand.
Hence, to illustrate how the reasoning method can be used for the verification of quantized GNNs, we slightly modify the LVP~$I$ with a different output constraint.

Let $I'$ be the LVP $(N, L_{\text{in}}, L'_{\text{out}}, \delta)$, where $N$ and $L_{\text{in}}$ are as before, $L'_{\text{out}} := y_1 \geq 0.9$, and $\delta = 5$.
We are going to show that the new output constraint is not always satisfied. 
$I'$ is valid iff $\phi_{I'} := \phi_{\text{in}} \land \phi_N \land \lnot \psi'_{\text{out}}$ is not satisfiable, where:
\begin{itemize}
    \item $\psi'_{\text{out}} := y_1 \geq 0.9$
\end{itemize}

\newcommand{\ruledef}{\Delta}

We employ our tableau method to prove that $\phi_{I'}$ is satisfiable and therefore that $I'$ is not valid.
We use $\checkmark$ to indicate that we have reached a term for which there are no rules to apply, except possibly $(clash_=)$. If in the end $(clash_=)$ does not apply, then we can read a model of $\phi_{I'}$ where the $\checkmark$s appear. 
We start the tableau method by applying the rule $(\land)$ and resolve inequalities with rules $(\geq)$ and $(\lnot \geq)$ (we write $(\ruledef)$ when we just use the definition of a formula):
\begin{prooftree}\small
\AxiomC{$(\epsilon ~ \phi_{\text{in}} \land \phi_N \land \lnot \psi'_{\text{out}})$}
\rulelabel{\land}
\UnaryInfC{$(\epsilon ~ \phi_{\text{in}}), (\epsilon ~ \phi_N), (\epsilon ~  \lnot \psi'_{\text{out}})$}
\end{prooftree}

\noindent
\begin{minipage}{0.15\textwidth}
\begin{prooftree}\scriptsize
\AxiomC{$(\epsilon ~ \phi_{\text{in}})$}
\rulelabel{\ruledef}
\UnaryInfC{$(\epsilon ~ x_1 {\geq} 100)$}
\rulelabel{\geq}
\UnaryInfC{$(\epsilon ~ x_1 {=} 100)$ $\checkmark$}
\end{prooftree}
\end{minipage}
     \begin{minipage}{0.15\textwidth}
     	\begin{prooftree}\scriptsize
     		\AxiomC{$(\epsilon ~ \phi_N)$}
     		\rulelabel{\ruledef}
     		\UnaryInfC{$(\epsilon ~ C_1 {\land} C_2)$}
     		\rulelabel{\land}
     		\UnaryInfC{$(\epsilon ~ C_1), (\epsilon ~ C_2)$}
     	\end{prooftree}
\end{minipage}
     \begin{minipage}{0.15\textwidth}
\begin{prooftree}\scriptsize
	\AxiomC{$(\epsilon ~ \lnot \psi'_{\text{out}})$}
	\rulelabel{\ruledef}
	\UnaryInfC{$(\epsilon ~ \lnot y_1{\geq} 0.9)$}
	\rulelabel{{\lnot} {\geq}}
	\UnaryInfC{$(\epsilon ~ y_1 {=} 0.8)$ 
		$\checkmark$}
\end{prooftree}
\end{minipage}

\noindent
where we write $C_1$ for the conjunct $ReLU(\frac{1}{125} x_1) - y_1 = 0$ and $C_2$ for the conjunct $\frac{1}{1000}\agg(x_1) - y_2 = 0$.


We handle the conjunct $C_1$ in $\phi_N$.
\begin{prooftree}\small
\AxiomC{$(\epsilon ~ ReLU(\frac{1}{125} x_1) + (-1) y_1 = 0)$}
\rulelabel{+}
\UnaryInfC{$
(\epsilon ~ ReLU(\frac{1}{125} x_1) = 0.8),
(\epsilon ~ (-1) y_1 = -0.8)
$ 
}
\end{prooftree}

\begin{minipage}{0.25\textwidth}
	\begin{prooftree}\small
\AxiomC{$
(\epsilon ~ ReLU(\frac{1}{125} x_1) = 0.8)
$
}
\rulelabel{\alpha}
\UnaryInfC{$(\epsilon ~ \frac{1}{125} x_1 = 0.8)$} 
\rulelabel{\times}
\UnaryInfC{$(\epsilon ~ x_1 = 100)$ $\checkmark$} 
\end{prooftree}
\end{minipage}
%
\begin{minipage}{0.15\textwidth}
\begin{prooftree}\small
\AxiomC{$
(\epsilon ~ (-1) y_1 = -0.8)
$
}
\rulelabel{\times}
\UnaryInfC{$(\epsilon ~ y_1 = 0.8)$ $\checkmark$} 
\end{prooftree}
\end{minipage}

Finally, we handle the conjunct $C_2$ in $\phi_N$.
\begin{prooftree}\small
\AxiomC{$(\epsilon ~ \frac{1}{1000}\agg(x_1) + (-1)y_2 = 0)$}
\rulelabel{+}
\UnaryInfC{$
(\epsilon ~ \frac{1}{1000}\agg(x_1) = 1),
(\epsilon ~ (-1)y_2 = -1)
$ 
}
\end{prooftree}

\begin{minipage}{0.25\textwidth}
\begin{prooftree}\small
\AxiomC{$
(\epsilon ~ \frac{1}{1000}\agg(x_1) = 1)
$
}
\rulelabel{\times}
\UnaryInfC{$(\epsilon ~ \agg(x_1) = 1000)$ 
}
\rulelabel{degree}
\UnaryInfC{$(\epsilon ~ degree = 4)$ $\checkmark$}
\end{prooftree}
\end{minipage}
\begin{minipage}{0.1\textwidth}
\begin{prooftree}\small
	\AxiomC{$
		(\epsilon ~ (-1)y_2 = -1)
		$
	}
	\rulelabel{\times}
	\UnaryInfC{$(\epsilon ~ y_2 = 1)$ 
		$\checkmark$}
\end{prooftree}
\end{minipage}

\begin{prooftree}\small
\AxiomC{$(\epsilon ~ \agg(x_1) = 1000)$}
\AxiomC{$(\epsilon ~ degree = 4)$}
\rulelabel{\agg}
\BinaryInfC{$(1 ~ x_1 = 250), \dotsc, (4 ~ x_1 = 250)$ 
$\checkmark$}
\end{prooftree}

\paragraph{Obtaining a model.}
We obtain a counterexample $(G,v)$ for the formula $\phi_{I'}$, with the node $v$ labelled $(x_1) = (100)$, which has $4$ successors, all labelled $(x_1) = (250)$. 
So, by \Cref{th:reduction}, $I'$ is not valid.
Indeed, $G,v$ satisfies $L_{in}$. But $N(G,v)$ is $(y_1, y_2) = (0.8, 1)$ which does not satisfy $L'_{out}$.

This illustrates how the proof system serves to verify the properties of quantized GNN, but also to exhibit counterexamples when they exist.

\section{Lower Bounds}


\label{sec:lowerbound}
\newcommand{\setofvaluesforPSPACEhardaritytwo}{\set{ 0, 1, 2}}
\newcommand{\setofvaluesforPSPACEhard}{\set{-2, -1, 0, 1, 2, 3}}

\newcommand{\casezerotwoatwo}{$\text{Case}_{\set{0,1,2}, \delta=2}$}
\newcommand{\caseminustwothree}{$\text{Case}_{\set{-2..3}, satu}$}

If not said otherwise, we suppose now the activation function is truncated ReLU; it corresponds to the settings in many works (\cite{DBLP:conf/iclr/BarceloKM0RS20,ijcai2024,benedikt2024decidability}). Furthermore, to obtain stronger PSPACE lower bound results, we will consider two restrictions:
\begin{itemize}
    \item \casezerotwoatwo: $\setnumbers$ contains $0$, $1$, $2$ 
    and the out-degree is bounded by $\degree{=}2$;
    \item \caseminustwothree: $\setnumbers$ includes $\setofvaluesforPSPACEhard$ and moreover, $\setnumbers$ saturates in the following sense: for all integers $x$, $x \geq 3$ implies $x_\setnumbers \geq_\setnumbers 3_\setnumbers$. In particular, it means that we have no modulo behavior in $\setnumbers$ (e.g., $\setnumbers = \mathbb Z_{10}$ would not saturate since despite $10 \geq 9$ as integers, we would not have $0 = 10 \geq_\setnumbers 9$ modulo $10$).
\end{itemize}
In both cases, $n=3$ bits are enough to represent the numbers.
Our lower bound results rely on \emph{graded modal logic} which extends modal logic with constructions $\gmlbox k \phi$ whose truth conditions is:
    $G, v \models \gmlbox k \phi$  iff there are at least $k$ successors $u$ of $v$ such that $G, u \models \phi$.
We call $\text{GML}_K$ the syntactic fragment 
where occurrences of $\gmlbox k$ are such that $k \leq K$.

\begin{lemma}
The satisfiability of $\fragmentGMLmodallogic$ on graphs wit out-degree at most $2$, and the satisfiability of $\fragmentGML$ is PSPACE-hard.
\label{lemma:GMLonegrapharitytwo}
\label{lemma:GMLthree}
\end{lemma}





\begin{theorem}
    LVP is PSPACE-complete. PSPACE-hardness holds already for \casezerotwoatwo\ or \caseminustwothree.
    \label{theorem:LVPPSPACEcomplete}
\end{theorem}

\begin{proof}
PSPACE-membership is stated in Corollary~\ref{cor:LVPPSPACE}. The lower bound relies on the poly-time transformation $\tau$ of a GML-formula $\phi$ into an equivalent GNN $\aGNN = \tau(\phi)$ in the sense that $\aGNN(G, v)$ satisfies $x_1 \geq 1$, given in 
\citeauthor{DBLP:conf/iclr/BarceloKM0RS20}~\shortcite[proof of Proposition~4.1]{DBLP:conf/iclr/BarceloKM0RS20}.
More specifically, we carefully analyze the GNN $N$ when the formula $\phi$ is in $\text{GML}_K$.

In $N$, the activation function is truncated ReLU. 
The layers in $\tau(\phi)$ are all the same. The matrices $C$ contain coefficients $-1$, $0$, $1$ and all coefficients on a column are $0$ except two at most.
The matrices $A$ only contain $0$ and $1$. On each column of $A$, all coefficients are $0$ except one which may be~$1$.
Coefficients in the biases $b$ are in $\set{ -K{+}1, -K{+}2, \dots, 1}$.

Propositional 
variables in $\phi$ are among $p_1, \dots, p_J$, 
$\aGNN(G, v)$ starts with the 
state 
$(\bs{x}^0_u)_j = 1$, if $j \in \set{1, \dots, J}$ and $p_j$ holds in $u$, and $(\bs{x}^0_u)_j = 0$ otherwise.
%
 As the activation function is truncated ReLU, all vectors $\bs{x}^i_u$ are $0$--$1$-vectors.

Given the coefficients in the matrices, intermediate values are all integers.
We end the proof of our Theorem~\ref{theorem:LVPPSPACEcomplete} by stating the reductions and arguing that the computation is \emph{safe} when done in~$\setnumbers$, meaning that there is no overflow and that the computation in~$\setnumbers$ and in 
$\mathbb{Z}$ yields the same result.
    
\newcommand{\LVPdual}{\ensuremath{\compl{\text{LVP}}}\xspace}

\paragraph{Proof for \casezerotwoatwo} 
We reduce the satisfiability of $\fragmentGMLmodallogic$ on graphs of out-degree at most $2$ which is PSPACE-hard (\Cref{lemma:GMLonegrapharitytwo}) into the dual problem \LVPdual of LVP with $\delta {=} 2$: from $\phi$ a $\fragmentGMLmodallogic$-formula, 
we compute the \LVPdual-instance $(\tau(\phi), \top, y_1 \leq 0, \degree{=}2)$, and $\phi$ is satisfiable on graphs of out-degree $2$ iff $(\tau(\phi), \top, y_1 \leq 0, \degree{=}2)$ is a negative instance of LVP. \LVPdual is PSPACE-hard, so is LVP.
The computation is safe because there are at most two successors, so we never exceed $2$ in a computation in $\setnumbers$.

\paragraph{Proof for \caseminustwothree} 
 We reduce the satisfiability of $\fragmentGML$ which is PSPACE-hard (\Cref{lemma:GMLthree}) into \LVPdual: given $\phi$ a $\fragmentGML$-formula, we compute the \LVPdual-instance $(\tau(\phi), \top, y_1 \leq 0, \delta{=}{+\infty})$.
We have that $\phi$ is $\fragmentGML$-satisfiable iff there exists $G, v$ s.t.\ $G, v \models \phi$ iff $(\tau(\phi), \top, y_1 \leq 0, \delta{=}{+\infty})$ is a negative instance of LVP. 
The computation is safe because, even if there are an arbitrary large number $x$ of successors, if $x \geq 3$, we also have $x_{\setnumbers} \geq_{\setnumbers}  3_{\setnumbers}$ because $\setnumbers$ saturates.
\end{proof}

\begin{corollary}
\label{corollary-satthelogicPSPACEc}
The satisfiability problem of \thelogic{} is PSPACE-complete, and already PSPACE-hard for \casezerotwoatwo\ or \caseminustwothree, and also when the activation function is ReLU.
\end{corollary}

\section{Related Work}
\label{sec:relatedwork}
Starting with Graded Modal Logic  (\cite{DBLP:journals/ndjfl/Fine72}),
there are numerous logics that capture modal aspects of graphs and express arithmetic constraints, 
(\cite{DBLP:journals/japll/DemriL10,DBLP:conf/frocos/Baader17,DBLP:conf/fsttcs/BednarczykOPT21,DBLP:conf/kr/GallianiKT23,DBLP:journals/bsl/BenthemI23}).

Previous research has already established several correspondences between logic and GNNs. For instance, \citet{DBLP:conf/iclr/BarceloKM0RS20} explore the relationship between graded modal logic and GNNs, while \citet{ijcai2024} examine a modal logic over linear inequalities with counting and its connection to GNNs. Additionally, \citet{benedikt2024decidability} investigates fragments of Presburger logic in the context of GNNs. However, these existing works focus on GNNs with specific activation functions and do not consider the broader class of quantized GNNs. In particular, decidability in PSPACE has been established only for cases where the activation function is either a truncated ReLU (\cite{ijcai2024}) or eventually constant functions (\cite{benedikt2024decidability}). 
\changed{\citet{AhvonenHKL24} offer several logical characterizations of recurrent GNNs over floats and real numbers.}
Also related, is the work of \citet{Grohe23}, which establishes that the graph queries computable by a polynomial-size, bounded-depth family of GNNs are precisely those definable in the guarded fragment 
of first-order logic with counting and built-in relations. 

\citet{HenzingerLZ21} also address the verification of quantized neural networks, but in contrast to this paper they focus on FNN. The authors establish a PSPACE-hardness result, relying on a binary representation of the number of bits. In contrast, if the number of bits were unary, their problem is presumably in NP (using a guess and check argument, relying on similar arguments as used by \cite{SalzerL21}).

There are other aggregation functions, like weighted sums in graph attention networks \cite{DBLP:conf/iclr/VelickovicCCRLB18}, maximum in Max-GNNs \cite{DBLP:conf/kr/CucalaG24}, or average. \thelogic{} and the tableau method can be adapted to capture these aggregation functions, too (see \citesuplementarymaterial).

On the practical side, \citet{DBLP:conf/aaai/000200DWZB24} presents a solution to the verification of quantized FNNs using heuristic search, and outperforming the approach of \citet{10.1145/3551349.3556916} based purely on integer linear programming.

\section{Conclusion and Perspectives}
\label{sec:outlook}

We introduced a method for explicitly verifying and reasoning about properties of a practical class of quantized Graph Neural Networks. 
It allows us to establish new foundational results about the computational complexity of GNNs' tasks.

We proposed a tableau method for reasoning about quantized GNNs using the new logic \thelogic{}. As the domain for numbers is finite, we are able to work out our details using a parameterized class of quantized GNNs, allowing us to consider various activation functions.
We showed that for such classes of GNNs using FNNs in their internal parts, we can answer basic formal reasoning questions using a reduction to the satisfiability problem of \thelogic{}.
We showed that satisfiability of  \thelogic{}
is solvable in PSPACE. 


This work is a first step to further develop concrete methods for reasoning about quantized GNNs and other neural network models. Now, it will be interesting to implement our approach on  a broader scale in order to tackle tasks such as certifying (adversarial) robustness properties (\cite{Gunnemann2022}) or giving formal explanations (\cite{0001I22_formalxai}) for certain GNN behaviours.

\section*{Ethical Statement}
There are no ethical issues.

\section*{Acknowledgments}

This work was supported by the ANR EpiRL project ANR22-CE23-0029, and by the MUR (Italy) Department of Excellence 2023–2027.


\bibliographystyle{named}
\bibliography{biblio}

\appendix

\section{Example of GNN With More Dimensions}

We present an example of GNN with more dimensions, and of a LVP instance, and provide all the details of an instance of computation. We can assume $\setnumbers$ to be the set of signed $4$-bit integers.

\paragraph{GNN $N$.}
Let $N = (\layer_1, \layer_2, \layer_{out})$ be a GNN, such that:
\begin{itemize}
    \item $\layer_1 = (\agg_1, \comb_1)$, $\agg_1 = \sum$, and $\comb_1((x_1, x_2)^T, (\aggvar_1, \aggvar_2)^T) = 
    \begin{pmatrix}
    ReLU(+2x_1 - 1x_2 + 1\aggvar_1 - 1\aggvar_2 - 1) \\
    ReLU(+1x_1 + 2x_2 - 2\aggvar_1 - 1\aggvar_2 + 1)
    \end{pmatrix}$
    \item $\layer_2 = (\agg_2, \comb_2)$,
    $\agg_2 = \sum$, and $\comb_2((x_1, x_2)^T, (\aggvar_1, \aggvar_2)^T) = 
    \begin{pmatrix}
    ReLU(+1x_1 - 1x_2 + 2\aggvar_1 + 2\aggvar_2 - 2) \\
    ReLU(+2x_1 + 0x_2 - 2\aggvar_1 - 1\aggvar_2 + 2) \\
    ReLU(-1x_1 + 1x_2 - 2\aggvar_1 + 1\aggvar_2 - 1)
    \end{pmatrix}$
\end{itemize}
Layer $\layer_1$ has input and output dimensions 2, while layer $\layer_2$ has input dimension 2, and output dimension 3. The output layer $\layer_{out}$ is the identity.

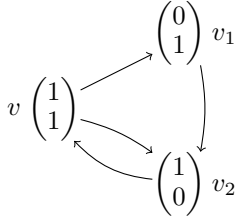
\begin{figure}[h]
    \centering
    	\begin{tikzpicture}
		\draw[fill=white, opacity=0.7,rounded corners = 5mm, draw=none] (-2.5, 1.5) rectangle (0.6, -1.5);
		\node[inner sep=0mm, rounded corners=2mm] (0) at (-2, 0) {\labellednode{v}{1}{1}};
		\node[inner sep=0mm] (1) at (0, 1) {\labellednodeinv{v_1}{0}{1}};
		\node[inner sep=0mm] (2) at (0, -1) {\labellednodeinv{v_2}{1}{0}};
		\draw[->] (0) -- (1);
		\draw[->] (0) edge [bend left=10] (2);
		\draw[->] (1) edge [bend left=10] (2);
		\draw[->] (2) edge [bend left=20] (0);
		
	\end{tikzpicture}
    \caption{The vector-labelled graph $G_e$.}
    \label{fig:graph2}
\end{figure}

\paragraph{Output of $N$ on input $G_e, v$.}
Consider the graph $G_e$ as represented in Figure~\ref{fig:graph2}.

Through the first layer,

$y^0(v) 
= \agg_1(\multiset{x^0(v_1), x^0(v_2)}) 
= \agg_2 (\multiset{(0, 1)^T, (1, 0)^T}) 
= (1,1)^T$

$x^1(v) 
= \comb_1(x^0(v), y^0(v)) 
= \comb_1((1, 1)^T, (1, 1)^T) 
=  \begin{pmatrix}
    ReLU(+2 - 1 + 1 - 1 - 1) \\
    ReLU(+1 + 2 - 2 - 1 + 1)
    \end{pmatrix} 
= (0, 1)^T$

$y^0(v_1) 
= \agg_1(\multiset{}) = (0, 0)^T$

$x^1(v_1) 
= \comb_1(x^0(v_1), \aggvar_0(v_1))
= \comb_1((0,1)^T, (0, 0)^T))
= \begin{pmatrix}
    ReLU(+0 - 1 + 0 - 0 - 1) \\
    ReLU(+0 + 2 - 0 - 0 + 1)
    \end{pmatrix}
= (0, 3)^T$

$y^0(v_2)
= \agg_1(\multiset{(1,1)^T}) = (1, 1)^T$

$x^1(v_2)
= \comb_1(x^0(v_2), \aggvar_0(v_2))
= \comb_1((1,0)^T, (1, 1)^T))
= \begin{pmatrix}
    ReLU(+2 - 0 + 1 - 1 - 1) \\
    ReLU(+1 + 0 - 2 - 1 + 1)
    \end{pmatrix}
= (1, 0)^T$

Through the second layer,

$y^1(v) 
= \agg_2(\multiset{x^1(v_1), x^1(v_2)})
= \agg_2(\multiset{(0,3)^T, (1,0)^T})
= (1,3)^T$

$x^2 (v) 
= \comb_2(x^1(v), y^1(v))
= \comb_2((0,1)^T, (1,3)^T)
= \begin{pmatrix}
    ReLU(+0 - 1 + 2 + 6 - 2) \\
    ReLU(+0 + 0 - 2 - 3 + 2) \\
    ReLU(-0 + 1 - 2 + 3 - 1)
    \end{pmatrix}
= (5,0,1)^T$


\paragraph{LVP instance.}
Let $I$ be the LVP instance $(N, L_{in}, L_{out})$, where $N$ is the GNN as defined before, $L_{in} := x_1 + x_2 \geq x_1 + x_2$, and $L_{out} := y_1 - y_2 - y_3 \geq 0$. (The linear inequality $L_{in}$ is tautological.)

$(G_e, v) = (1,1)^T$ satisfies $L_{in}$ and $N(G_e,v) = (5,0,1)^T$ satisfies $L_{out}$.

But $I$ is not valid.
To see that, consider the graph $G$ that has only one vertex $v$, labelled $\ell(v) = \begin{pmatrix}
        0\\0
\end{pmatrix}$.

Through the first layer,

$y^0(v) 
= (0,0)^T$

$x^1(v)
=  \begin{pmatrix}
    ReLU(- 1) \\
    ReLU(+ 1)
    \end{pmatrix} 
= (0, 1)^T$

Through the second layer,

$y^1(v) 
= (0,0)^T$

$x^2 (v) 
= \begin{pmatrix}
    ReLU(- 1 - 2) \\
    ReLU(+ 0 + 2) \\
    ReLU(+ 1 - 1)
    \end{pmatrix}
= (0,2,0)^T$

$(G, v) = (0,0)^T$ satisfies $L_{in}$ but $N(G,v) = (0,2,0)^T$ does not satisfy $L_{out}$.

\section{Omitted Proofs}

\subsection{Proof of \Cref{th:satlogPSPACE} (Out-Degree in Binary or Unbounded)}
\label{sec:proofs-arity-binary-or-unbounded}

In this subsection, we prove that \Cref{th:satlogPSPACE} holds when the out-degree is written in binary or is unbounded ($\degree = +\infty$). The unbounded case still means that the branching is finite at each node but unbounded. We prove PSPACE membership by adapting the tableau rules given in Figure~\ref{figure:tableaurules}. The new tableau rules will work both for the cases when the out-degree is written in binary and the case when it is unbounded ($\degree = +\infty$).

 \paragraph{Rule for guessing the out-degree.} First, we modify the rule $(degree)$ as follows:

	\begin{prooftree}
		\AxiomC{$(w~\phi)$ and no term $(w~degree=..)$}
		\rulelabel{degree}
		\UnaryInfC{$(w~degree=\delta')$ for some $\delta'\leq \delta$ with $\delta' \leq 2^n |\phi|$}
	\end{prooftree}
	where this time $\delta'$ is written in binary. The rule $(degree)$ works when $\delta$ is an integer or $+\infty$.
In the above rule, we guess the out-degree of a node $w$ to be bounded by $\delta$ like in the rule in~\Cref{figure:tableaurules}. But, we additionally suppose that $\delta'$ is an integer between $0$ and $2^n |\phi|$ written with $O(n \log |\phi|)$ bits. Why is the bound $2^n |\phi|$ sufficient? The idea is that having $\delta'$ greater than $2^n |\phi|$ is useless because $2^n |\phi|$ is a bound for having all the combinations for $\expression_1 = k_1, \dots, \expression_m = k_m$ for $\agg(\expression_1), \dots, \agg(\expression_m)$ appearing in $w$ (we thus have $m \leq |\phi|)$.

\paragraph{Rules for aggregation. }
Whereas there is a single rule $(\agg)$ in Figure~\ref{figure:tableaurules} for handling the aggregation, we will introduce three rules.

	First, we add the following rule that initializes the computation for $\agg(\expression) = k$:
	
	\begin{prooftree}
		\AxiomC{$(w~\agg(\expression) = k)$}
		\rulelabel{\agg_{init}}
		\UnaryInfC{
			$(w~~0~~\agg(\expression) = k~~0)$
		}
	\end{prooftree}
	
	The term $(w~~i~~\agg(\expression) = k~~A)$ indicates that we have already performed the computation for the successors $w1, \dots, wi$ and the cumulative sum so far is $A$. With the rule $(\agg_{init})$, it means that we did not consider any successors so far and the cumulative sum so far is zero.

The next rule $(agg_{rec})$ is the \emph{recursive} rule and treats the $i$-th successor as follows. It guesses the value $k'$ for feature $\expression$ for the successor $w(i+1)$. It updates the cumulative result so far to $A+k'$.
	
	\begin{prooftree}\small
		\AxiomC{\begin{minipage}{0.8\linewidth}$(w~~i~~\agg(\expression) = k ~~A)~~ (w~~degree=\delta')$ with $i<\delta'$
  \end{minipage}}
		\rulelabel{agg_{rec}}
		\UnaryInfC{\begin{minipage}{0.8\linewidth}
			$(w(i+1)~~\expression {=} k')~~(w~~i{+}1~~\agg(\expression) = k~~A+k')$ for some $k'$
\end{minipage}		
  }
	\end{prooftree}

 The last rule $(agg_{fail})$ says that the tableau method fails if the cumulative value $A$ obtained when all the successors $w1, \dots, w\delta'$ have been considered is different from the expected $k$:

	\begin{prooftree}\small
		\AxiomC{\begin{minipage}{0.8\linewidth}
  $(w~~degree = \delta')~~ (w~~\delta'~~\agg(\expression) = k~~A)$ with $A \neq k$
  \end{minipage}
  }
		\rulelabel{agg_{fail}}
		\UnaryInfC{fail
		}
	\end{prooftree}

	The procedure can still be executed in PSPACE by exploring the successors in a depth-first-search manner. We treat each successor one after another. After treating $wi$ completely, we free the memory for the whole computation for successor $wi$.

\subsection{Proof of \Cref{lemma:GMLonegrapharitytwo}}

Before proving \Cref{lemma:GMLonegrapharitytwo} we focus on modal logic over graphs of out-degree at most $2$.

\subsubsection{Modal logic over graphs of out-degree at most $2$}

Similarly to standard modal logic, the satisfiability problem of modal logic over the class of graphs of out-degree at most $2$, that we denote by $\Ktwoa$, is also PSPACE-hard. Its proof is similar to the one for standard modal logic in \citet{DBLP:books/cu/BlackburnRV01}.

\begin{lemma}
	\label{lemma:K2a}
    The satisfiability problem of $\Ktwoa$ is PSPACE-hard.
\end{lemma}

\begin{proof}
	By poly-time reduction from TQBF. As said, the proof follows the proof of PSPACE-hardness of modal logic K \cite{DBLP:books/cu/BlackburnRV01}. We give the main argument in order to be self-contained.
	
	Let us take a quantified binary formula $\exists p_1 \forall p_2 \dots \exists p_{2n-1} \forall p_{2n} \chi$ where $\chi$ is propositional.
	The game behind TQBF can be represented by the binary tree in which $p_i$ is chosen at the $i$-th level.

	Modal logic, even restricted to models of out-degree bounded by $2$, can express that the Kripke model contains this binary tree. The following formula explains how the $i$-th level should look like:
 $$\lbox ^i \left( \ldiamond p_i \land \ldiamond \lnot p_i \land \lbigand_{j<i} (p_j \lequiv \lbox p_j) \right)$$

	Let $TREE$ be the conjunction of the above formulas for $i=1..2n$ that forces the Kripke model to contain the binary tree up to level $2n$ (see Figure~\ref{fig:binarytreeforTQBF}). On the input $\exists p_1 \forall p_2 \dots \exists p_{2n-1} \forall p_{2n} \chi$, the reduction computes in polynomial time the modal formula $TREE \land \ldiamond \lbox \dots \ldiamond \lbox \chi$. The former is true iff the latter is satisfiable.

	\begin{figure*}[ht]
		\begin{center}
			\scalebox{0.8}{
				\begin{tikzpicture}[yscale=0.4, xscale=1, level distance=15mm,
					level 4/.style={sibling distance=14mm},
					level 3/.style={sibling distance=26mm},
					level 2/.style={sibling distance=50mm},
					level 1/.style={sibling distance=100mm},
					]
					\node {}
					child { node {$p_1$} 
						child { node {$p_1p_2$}
							child {node {$p_1p_2p_3$}
								child {node {$p_1p_2p_3p_4$}}
								child {node {$p_1p_2p_3$}}
							}
							child {node {$p_1p_2$}
								child {node {$p_1p_2p_4$}}
								child {node {$p_1p_2$}}
							}
						}
						child { node {$p_1$}
							child {node {$p_1p_3$}
								child {node {$p_1p_3p_4$}}
								child {node {$p_1p_3$}}
							}
							child {node {$p_1$}
								child {node {$p_1p_4$}}
								child {node {$p_1$}}
							}
						}
					}
					child { node {} 
						child { node {$p_2$}
							child {node {$p_2p_3$}
								child {node {$p_2p_3p_4$}}
								child {node {$p_2p_3$}}
							}
							child {node {$p_2$}
								child {node {$p_2p_4$}}
								child {node {$p_2$}}
							}
						}
						child { node {}
							child {node {$p_3$}
								child {node {$p_3p_4$}}
								child {node {$p_3$}}
							}
							child {node {}
								child {node {$p_4$}}
								child {node {}}
							}
						}
					};
			\end{tikzpicture}}
		\end{center}
		
		\caption{Binary tree for a TQBF game for $n=2$ (i.e., variables $p_1, p_2, p_3, p_4$. The value of $p_1$ is chosen at the root. The value of $p_2$ is chosen at depth 1, etc.}
		\label{fig:binarytreeforTQBF}
	\end{figure*}
	
\end{proof}

\subsubsection{Proof of \Cref{lemma:GMLonegrapharitytwo}}


\begin{proof}
\fbox{$\fragmentGMLmodallogic$ on graphs of out-degree at most 2}
We reduce the satisfiability problem of $\Ktwoa$ (which is PSPACE-hard, see \Cref{lemma:K2a}) to the satisfiability problem of $\fragmentGMLmodallogic$ on graphs of out-degree at most 2.
We translate modal logic into GML as follows: $tr(p) = p$ and $tr(\lbox \phi) = \lnot \gmlbox 1 \lnot tr(\phi)$. In other words, $\lbox \phi$ is translated as `it is false that there is no successor with $tr(\phi)$ false'.

\fbox{$\fragmentGML$}
We reduce the satisfiability problem of $\Ktwoa$ (which is PSPACE-hard, see \Cref{lemma:K2a}) to the satisfiability problem of $\fragmentGML$.
We can prove that  we can express that the out-degree is bounded by $2$ in $\fragmentGML$.

	We then define $\tau(\phi) = tr(\phi) \land \lbigand_{i=0}^d (\lnot \gmlbox 1\lnot )^i \lnot \gmlbox 3 \top$
    where $d$ is the modal depth of $\phi$, and where $tr$ is the translation given in the proof of \Cref{lemma:GMLonegrapharitytwo}. 
    
    The subformula $\lnot \gmlbox 3 \top$ says that it is false that there are at least~3 successors.
	The subformula $\lbigand_{i=0}^d (\lnot \gmlbox 1\lnot )^i \lnot \gmlbox 3 \top$ says that the vertices should have at most 2 successors, up to the modal depth of $\phi$.
	
	Now it can be proven by induction on $\phi$ that 
	$\phi$ is $\Ktwoa$-satisfiable iff $tr(\phi)$ is $\fragmentGML$-satisfiable.
\end{proof}

\subsection{Proof of \Cref{corollary-satthelogicPSPACEc}}

\begin{proof}
      PSPACE membership is stated in Theorem~\ref{th:satlogPSPACE}. Hardness is proven by the poly-time reduction (see Theorem~\ref{th:reduction}) from LVP (which is PSPACE-hard in this context, see \Cref{theorem:LVPPSPACEcomplete}) to the satisfiability problem of~\thelogic{}. 
\end{proof}

\begin{corollary}
The satisfiability problem of \thelogic{} is PSPACE-complete, and already PSPACE-hard in the restrictive cases and if the activation function is ReLU.
\end{corollary}

\begin{proof}
    It suffices to observe that truncated ReLU can be simulated as follows:
    $truncatedReLU(x) := RelU(ReLU(x) - ReLU(x-1))$.
    So we can reduce the satisfiability problem of \thelogic{} with truncated ReLU to the satisfiability problem of \thelogic{} with ReLU. To do that, we replace each occurrence of $\activationfunction(E)$ (where $\activationfunction$ was previously interpreted as truncated ReLU) by $\activationfunction(\activationfunction(E) - \activationfunction(E-1))$ (where now $\activationfunction$ is interpreted as ReLU). This operation does not blow up the size of the formula since formulas are represented by DAGs. 
\end{proof}

\section{Discussion on Other Aggregation Functions}
\label{appendix-section-aggregation}

In this section, we show how our methodology can be easily applied for other aggregation functions, like those, for instance, compared in \citet{RosenbluthTG23}. The idea is to replace the $\agg$ operator.
\paragraph{Weighted sum in aggregation functions}

In variants of GNN, the summation in the aggregation function is weighted (e.g. \cite{DBLP:conf/iclr/VelickovicCCRLB18}). Assuming that the weights are fixed for each edge, we can replace operator $\agg$ with $\agg_{\textsf w_1\dots \textsf w_\delta}$ where  $\textsf w_1\dots \textsf w_\delta$ are weights. The rule will then have the form:
\begin{prooftree}
\AxiomC{$(w~\agg_{\textsf w_1\dots \textsf w_\delta}(\expression) = k)$}
\AxiomC{$(w~degree=\delta')$}
\rulelabel{\agg_{\vec{\textsf w}}}
\BinaryInfC{
\begin{minipage}{0.4\textwidth}
$(w1~\expression = k_{1}), \dotsc, (w \degree'~\expression = k_{{\degree'}})$ for some $(k_u)_{u=1..\degree'}$, with $\textsf w_1 k_{1} +_{\setnumbers_n} 
 \dotsb +_{\setnumbers_n} \textsf w_{\degree'} k_{\degree'} = k$
\end{minipage}
}
\end{prooftree}

\paragraph{Max-GNNs}

In Max-GNNs (\cite{DBLP:conf/kr/CucalaG24}), the aggregation function is the maximum. Our framework can be adapted by replacing the operator $\agg$ with $\max$. We then use the following rule:
\begin{prooftree}
\AxiomC{$(w~\max(\expression) = k)$}
\AxiomC{$(w~degree=\delta')$}
\rulelabel{\max}
\BinaryInfC{
\begin{minipage}{0.4\textwidth}
$(w1~\expression = k_{1}), \dotsc, (w \degree'~\expression = k_{{\degree'}})$ for some $(k_u)_{u=1..\degree'}$, with $\max(k_{1},
 \dotsb, k_{\degree'}) = k$
\end{minipage}
}
\end{prooftree}

\paragraph{Mean-GNNs}

In Mean-GNNs (\cite{RosenbluthTG23}), the aggregation function is the average. We can adapt our framework replacing the operator $\agg$ with $mean$.
It can scale better than sum aggregation across nodes with varying out-degrees. 
The behavior of this aggregation function is captured by the following rule:
\begin{prooftree}
\AxiomC{$(w~mean(\expression) = k)$}
\AxiomC{$(w~degree=\delta')$}
\rulelabel{mean}
\BinaryInfC{
\begin{minipage}{0.4\textwidth}
$(w1~\expression = k_{1}), \dotsc, (w \degree'~\expression = k_{{\degree'}})$ for some $(k_u)_{u=1..\degree'}$, with $(k_{1} +_{\setnumbers_n} 
 \dotsb +_{\setnumbers_n} k_{\degree'}) /_{\setnumbers_n} \degree' = k$
\end{minipage}
}
\end{prooftree}

\medskip
In every case, this yields a PSPACE decision procedures for satisfiability and validity checking. In the case of weighted sum, this complexity is tight.

\section{Implementation}
\label{appendix-section-implementation}

\newcommand{\implnumbers}[1]{\mathbb{Z}_{[-#1,#1]}}

We provide a Python implementation\footnote{\url{https://github.com/francoisschwarzentruber/ijcai2025-verifquantgnn}} of our tableau method where the set of numbers is of the form $\implnumbers{a} := \set{-a, -a+1, \dots, 0, 1, \dots, a-1, a}$, with the amplitude $a \in \mathbb{N}$, and where overflows are handled as follows: $x+_{\implnumbers{a}} y = \max(-a, \min(a, x+y))$.
We also assume $\degree = 2$.

The implementation relies on pattern matching available in Python 3.10 for implementing the tableau rules.

Formulas are written by means of Python tuples. Here is an example of an input formula \texttt{phi}:

\begin{scriptsize}
\begin{verbatim}
phiIn = (("agg", "x1"), ">=", 1)
phiN = (
    ((("ReLU", ("*", 2, "x1")), "-", "y1"), "=", 0),
    "and",
    ((("*", 2, ("agg", "x1")), "-", "y2"), "=", 0)
)
phiOut = ("y1", ">=", 1)
phi = ((phiIn, "and", phiN), "and", ("not", phiOut))
\end{verbatim}
\end{scriptsize}

The program starts with a single set of terms:

\begin{scriptsize}
\begin{verbatim}
{(1, phi)}
\end{verbatim}
\end{scriptsize}

Then the program applies the rules:

\begin{scriptsize}
\begin{verbatim}
initialization:  {(1, (((('agg', 'x1'), '>=', 1), 
'and', (((('ReLU', ('*', 2, 'x1')), '-', 'y1'), '=', 0),
'and', ((('*', 2, ('agg', 'x1')), '-', 'y2'), '=', 0))),
'and', ('not', ('y1', '>=', 1))))}
apply ruleAnd : {(1, ((('agg', 'x1'), '>=', 1), 'and', 
(((('ReLU', ('*', 2, 'x1')), '-', 'y1'), '=', 0), 'and',
((('*', 2, ('agg', 'x1')), '-', 'y2'), '=', 0)))),
(1, ('not', ('y1', '>=', 1)))}
apply ruleAnd : {(1, (('agg', 'x1'), '>=', 1)),
(1, (((('ReLU', ('*', 2, 'x1')), '-', 'y1'), '=', 0), 
'and', ((('*', 2, ('agg', 'x1')), '-', 'y2'), '=', 0))),
(1, ('not', ('y1', '>=', 1)))}
apply ruleAnd : {(1, (('agg', 'x1'), '>=', 1)),
(1, ((('ReLU', ('*', 2, 'x1')), '-', 'y1'), '=', 0)),
(1, ((('*', 2, ('agg', 'x1')), '-', 'y2'), '=', 0)),
(1, ('not', ('y1', '>=', 1)))}
apply ruleNot : 
{(1, ((('*', 2, ('agg', 'x1')), '-', 'y2'), '=', 0)), 
(1, (('*', -1, 'y1'), '>=', -1)),
(1, ((('ReLU', ('*', 2, 'x1')), '-', 'y1'), '=', 0)),
(1, (('agg', 'x1'), '>=', 1))}
new branch:
{(1, ((('*', 2, ('agg', 'x1')), '-', 'y2'), '=', 0)),
(1, (('*', -1, 'y1'), '=', -1)), 
(1, ((('ReLU', ('*', 2, 'x1')), '-', 'y1'), '=', 0)), 
(1, (('agg', 'x1'), '>=', 1))}
new branch:  
{(1, ((('*', 2, ('agg', 'x1')), '-', 'y2'), '=', 0)),
(1, ((('ReLU', ('*', 2, 'x1')), '-', 'y1'), '=', 0)),
(1, (('*', -1, 'y1'), '=', 0)), (1, (('agg', 'x1'), '>=', 1))}
...
\end{verbatim}
\end{scriptsize}

On that example, the program stops by giving a set of terms:

\begin{scriptsize}
\begin{verbatim}
no more rules applicable
example of final tableau: 
{(12, ('x1', '=', 4)),
(11, ('x1', '=', -3)),
(1, ('x1', '=', -2)),
(1, ('y1', '=', 0)),
(1, ('y2', '=', 2))}
\end{verbatim}
\end{scriptsize}
from which we get a pointed labelled graph
\begin{center}
    \begin{tikzpicture}
        \node (1) {$1: x1 = -2, y1 = 0, y2 = 2$};
        \node (11) at (-2, -2) {$11: x1 = -3$};
        \node (12) at (2, -2) {$12: x1 = 4$};
        \draw (1) edge[->] (11);
        \draw (1) edge[->] (12);
    \end{tikzpicture}
\end{center}
that satisfies the input formula.

We report a basic performance evaluation on a MacBook Apple M3 Pro with 36GB memory. It displays the time for running the tableau method starting with the formula \texttt{phi}, for varying values of $a$, hence using numbers in $\implnumbers{a}$.

\medskip

\begin{tikzpicture}[]
    \begin{axis}[
        xlabel={Value of the amplitude $a$},
        ylabel={Total time (seconds)},
        grid=both,
        grid style={dashed, gray!30},
        ymin=0,
        xmin=1,
        xtick={1,2,3,4,5,6,7,8,9,10,11,12,13,14,15,16,17,18,19,20},
        xticklabel style={rotate=45, anchor=east},
        ticklabel style={font=\tiny},
        width=8cm,
        height=5cm,
        title={}
    ]
        \addplot[
            thick,
            color=blue,
            mark=*,
            mark options={solid, fill=blue}
        ]
        coordinates {
            (1, 0.040)
            (2, 0.042)
            (3, 0.048)
            (4, 0.064)
            (5, 0.088)
            (6, 0.201)
            (7, 0.299)
            (8, 0.900)
            (9, 1.290)
            (10, 3.345)
            (11, 4.978)
            (12, 11.353)
            (13, 15.012)
            (14, 32.158)
            (15, 40.698)
            (16, 81.39)
            (17, 98.33)
            (18, 167.00)
            (19, 195.67)
            (20, 342.18)
        };
    \end{axis}
\end{tikzpicture}


\end{document}